\documentclass[draftcls,onecolumn]{IEEEtran}

\usepackage[utf8]{inputenc}
\usepackage[T1]{fontenc}
\usepackage{url}
\usepackage{ifthen}
\usepackage{cite}
\usepackage[cmex10]{amsmath}
\usepackage{amssymb}
\usepackage{amsthm}
\usepackage{mathtools}
\usepackage{bbm}
\usepackage{xcolor}

\usepackage{algorithm}
\usepackage{algpseudocode}

\newtheorem{definition}{Definition}
\newtheorem{theorem}{Theorem}
\newtheorem{proposition}{Proposition}
\newtheorem{lemma}{Lemma}
\newtheorem{corollary}{Corollary}
\newtheorem{remark}{Remark}

\begin{document}

\title{Lower Bounds for the MMSE via Neural Network Estimation and Their Applications to Privacy}
\author{Mario~Diaz,~\IEEEmembership{Member,~IEEE,}
Peter~Kairouz,~\IEEEmembership{Member,~IEEE,} \and
and~Lalitha~Sankar,~\IEEEmembership{Senior Member,~IEEE}
\thanks{The work of M.~Diaz was supported in part by the Programa de Apoyo a Proyectos de Investigaci\'{o}n e Innovaci\'{o}n Tecnol\'{o}gica (PAPIIT) under grant IA101021.
The work of L.~Sankar is supported in part by National Science Foundation grants CIF-1901243, CIF-1815361, and CIF-2007688.
This paper was presented in part at the 2021 IEEE International Symposium on Information Theory \cite{diaz2021neural} and at the Workshop on Information-Theoretic Methods for Rigorous, Responsible, and Reliable Machine Learning within the Thirty-eighth International Conference on Machine Learning.}
\thanks{M.~Diaz is with the Instituto de Investigaciones en Matem\'{a}ticas Aplicadas y en Sistemas (IIMAS), Universidad Nacional Aut\'{o}noma de M\'{e}xico, Mexico City 04510, Mexico (e-mail: mario.diaz@sigma.iimas.unam.mx).}
\thanks{P.~Kairouz is with Google Research, Seattle, WA, USA (e-mail: kairouz@google.com).}
\thanks{L.~Sankar is with Arizona State University, Tempe, AZ 85287 USA (e-mail: lsankar@asu.edu).}}

\maketitle

\begin{abstract}
The minimum mean-square error (MMSE) achievable by optimal estimation of a random variable $Y\in\mathbb{R}$ given another random variable $X\in\mathbb{R}^{d}$ is of much interest in a variety of statistical settings. In the context of estimation-theoretic privacy, the MMSE has been proposed as an information leakage measure that captures the ability of an adversary in estimating $Y$ upon observing $X$. In this paper we establish provable lower bounds for the MMSE based on a two-layer neural network estimator of the MMSE and the Barron constant of an appropriate function of the conditional expectation of $Y$ given $X$. Furthermore, we derive a general upper bound for the Barron constant that, when $X\in\mathbb{R}$ is post-processed by the additive Gaussian mechanism and $Y$ is binary, produces order optimal estimates in the large noise regime. In order to obtain numerical lower bounds for the MMSE in some concrete applications, we introduce an efficient optimization process that approximates the value of the proposed neural network estimator. Overall, we provide an effective machinery to obtain provable lower bounds for the MMSE.
\end{abstract}

\section{Introduction}

The disclosure of individual data could pose severe privacy risks \cite{sweeney2015only}. Even when the data being disclosed is not necessarily private, it could be correlated with sensitive information creating privacy vulnerabilities. To reduce risks, a widely adopted solution is to use a \emph{privacy mechanism} to sanitize the non-private data prior to its disclosure, see, e.g., \cite{dwork2014algorithmic,asoodeh2016information}. In precise terms, we have a Markov chain $Y - X - \tilde{X}$ where $Y$ represents private information (e.g., gender), $X$ represents non-private information (e.g., height), and $\tilde{X}$ is a noisy version of $X$. Indeed, a common privacy mechanism for continuous non-private data is the so-called additive Gaussian mechanism that adds an independent Gaussian random variable $Z$ to the non-private data, i.e., $\tilde{X} = X + Z$.

A prominent property of privacy is that it could be impossible to restore once it has been breached. For example, if an individual's released data leads to the inference that he has a chronic disease, it is impossible to restore the privacy of his condition once this fact has been exposed. Hence, it is important to know beforehand the privacy risks associated with the release of data. To this end, many measures of information leakage have been proposed in the literature. Typically, information leakage measures are formulated to capture the ability of an \emph{adversary} to make specific inferences about $Y$ upon observing $\tilde{X}$. For example, in the context of estimation-theoretic privacy, the minimum mean-square error (MMSE) in estimating $Y$ given $\tilde{X}$ captures the ability of the \emph{strongest} adversary aiming to approximate $Y$ in the expected square-loss sense \cite{asoodeh2016privacy,asoodeh2016information}. In this context, privacy guarantees naturally come in the form of lower bounds for the MMSE, as they ensure that such an adversary cannot estimate $Y$ beyond a certain precision.

Given a privacy mechanism, it could be challenging to predict its performance in practical applications. Oftentimes, theoretical guarantees are obtained through worst-case analyses that result in loose bounds. There is an active research area that aims to overcome the latter challenge by analyzing, designing, and auditing privacy mechanisms in a data-driven manner. One way to implement this philosophy is by estimating the information leakage of $Y$ in $\tilde{X}$ based on samples of these random variables \cite{shamir2010learning,jiao2015minimax,wu2016minimax,diaz2019robustness}. In the context of estimation-theoretic privacy, this amounts to establishing empirical lower bounds for the MMSE in estimating $Y$ given $\tilde{X}$.

In this work, we derive provable lower bounds for the MMSE in estimating $Y\in\mathbb{R}$ given\footnote{For ease of notation, we drop the tilde in $\tilde{X}$ and denote it just by $X$.} $X\in\mathbb{R}^{d}$. These lower bounds are based on a neural network estimator of the MMSE and the Barron constant of (a function of) the conditional expectation of $Y$ given $X$. More specifically, we propose the minimum empirical square-loss attained by a two-layer neural network as an estimator of the MMSE. Furthermore, we derive a general upper bound for the Barron constant that, when $X\in\mathbb{R}$ is post-processed by the additive Gaussian mechanism and $Y$ is binary, produces order optimal estimates in the large noise regime. In order to obtain numerical lower bounds for the MMSE in some concrete applications, we also analyze some algorithmic aspects related to the computation of the proposed estimator of the MMSE and introduce an efficient optimization process to approximate it.

The rest of the paper is organized as follows. In the remainder of this section we discuss further related work and recall some common notation used through this paper. In Section~\ref{Section:ProblemSettingPreliminaries} we present some elements of estimation-theoretic privacy, introduce our proposed estimator, recall Barron's approximation theorem, and discuss some aspects of the additive Gaussian mechanism. We derive lower bounds for the MMSE upon the proposed estimator in Section~\ref{Section:NNEstimation}. In Section~\ref{Section:GeneralBoundBarronConstant}, we derive a general bound for the Barron constant which, in Section~\ref{Section:BarronConstantPostProcessing}, yields order optimal estimates in the presence of the additive Gaussian mechanism. In Section~\ref{Section:NumericalConsiderations}, we consider some numerical aspects related to the computation of the proposed estimator of the MMSE and instantiate our lower bounds in a particular example. We provide a summary and some final remarks in Section~\ref{Section:Conclusion}.

\textbf{Related Work.} The minimum mean-square error achievable by optimal estimation of a random variable given another one plays a key role in statistics and communications \cite{scharf1991statistical,biglieri2007mimo}, and it is closely related to fundamental information-theoretic concepts \cite{guo2005mutual}. As such, the problem addressed in this paper is related to other fundamental problems in information theory and statistics. In the special case when $Y = X$, our setting is closely related to the problem of estimation in Gaussian channels as studied in \cite{guo2005mutual,guo2011estimation,wu2011functional}. Indeed, the problem considered in this work generalizes the aforementioned problem by considering finite samples and $Y \neq X$.

There exists a vast literature on MMSE estimation techniques, including those relying on linear \cite{scharf1991statistical}, kernel-based \cite{peinado2017statistical} and polynomial \cite{alghamdi2019mutual,alghamdi2021polynomial} approximations of the conditional expectation. In line with the recent surge in neural network estimation methods for information measures, see, e.g., \cite{belghazi2018mutual,chan2019neural,sreekumar2021non}, we adopt neural network estimation in the context of the MMSE. It is important to remark that while neural network estimation is known to underperform in some settings, see, e.g., \cite{mcallester2020formal}, our work is aligned with the theoretical nature of privacy where quantitative guarantees for the proposed methodologies are fundamental. This contrasts with existing MMSE estimation methodologies that focus mainly on empirical performance or provide only qualitative guarantees (e.g., convergence rates with unspecified constants).

There are several notions of privacy designed to capture the risks posed by a variety of adversaries, e.g., differential privacy quantifies the membership inference capabilities of an adversary in the context of database queries \cite{kairouz2017composition}, maximal $\alpha$-leakage quantifies the capacity of an adversary to infer any (randomized) function of the private attribute \cite{issa2017operational,liao2019tunable}, probability of correctly guessing quantifies the probability of an adversary to guess the private attribute \cite{asoodeh2018estimation}, to name a few notions. As mentioned before, our work belongs to research area dedicated to the data-driven estimation of information (leakage) measures, see, e.g., \cite{shamir2010learning,jiao2015minimax,wu2016minimax,diaz2019robustness} and references therein. There is also a recent research effort dedicated to understanding the performance of an adversary with practical computational capabilities \cite{jagielski2020auditing,nasr2021adversary}. From this perspective, our results compare the performance of a finite capacity adversary using a 2-layer neural network and finitely many samples to the performance of the strongest adversary capable of implementing any function and knowing the joint distribution of the private and disclosed data.

At a technical level, our starting point is Barron's approximation theorem \cite{barron1993universal}. While there is a variety of works extending Barron's result, see, e.g., \cite{breiman1993hinging,barron1994approximation,lee2017ability,klusowski2018approximation,ongie2019function,sreekumar2021non,domingo2021tighter}, most of them are asymptotic analyses in which the Barron constant is a fixed, yet unknown quantity. In contrast, in the present paper we show that this constant can be effectively controlled in the context of privacy under the additive Gaussian mechanism. To the best of the authors' knowledge, this is the first time that a quantitative analysis of the Barron constant is performed at the proposed level of generality.

\textbf{Notation.} We let $(\Omega,\mathcal{F},\mathbb{P})$ be the underlying probability space and $\mathbb{E}$ be the corresponding expectation. We denote by $\mathbbm{1}_{E}$ the indicator function of any set $E\in\mathcal{F}$. We let $\mathrm{Unif}(\mathcal{U})$ be the uniform distribution over $\mathcal{U}$. If $f:\mathbb{R}^{d}\to\mathbb{R}$ is a probability density function, we let $\mathrm{Supp}(f)$ be its support. If $p\in[0,1]$, we let $\bar{p} = 1 - p$. For $u,v\in\mathbb{R}^{d}$, we let $u \cdot v = u_{1}v_{1} + \cdots + u_{d}v_{d}$ and $\lvert u \rvert = \sqrt{u \cdot u}$. Unless otherwise stated, we let $\lVert \cdot \rVert_{p}$ be the $p$-norm in $L^{p}(\mathbb{R})$ or $L^{p}(\mathbb{R}^{d})$, depending on the context. We say that $f:\mathbb{R}^{d}\to\mathbb{R}$ is $\rho$-Lipschitz if $\lvert f(u) - f(v) \rvert \leq \rho \lvert u - v \rvert$ for all $u,v\in\mathbb{R}^{d}$. Also, we say that $f$ is of class $C^{m}$ if it has continuous partial derivatives of order up to $m$. We write $f(z) \sim g(z)$ to denote that $f(z)/g(z) \to 1$ as $z\to\infty$. We let $\tanh:\mathbb{R}\to(-1,1)$ be the hyperbolic tangent function. For $B\subset\mathbb{R}^{d}$, we let $\mathrm{rad}(B) \coloneqq \sup_{x\in B} \lvert x \rvert$. Recall that the gamma function is determined by $\Gamma(z) = \int_{0}^{\infty} x^{z-1} e^{-x} \mathrm{d}x$ for $z>0$.

\section{Problem Setting and Preliminaries}
\label{Section:ProblemSettingPreliminaries}

In this section we review some preliminary material on estimation-theoretic privacy, function approximation capabilities of two-layer neural networks, and the so-called additive Gaussian mechanism. Also, we introduce a neural network estimator of the MMSE that is used to derive the theoretical lower bounds for the MMSE in Section~\ref{Section:NNEstimation}.

\subsection{Estimation-Theoretic Privacy}

Given random variables $X\in\mathbb{R}^{d}$ and $Y\in\mathbb{R}$, the minimum mean square error in estimating $Y$ given $X$ is defined as\footnote{Observe that this definition requires to have random variables with finite second moments. Since in this paper we always deal with bounded random variables, this requirement is immaterial.}
\begin{equation}
\label{eq:DefMMSE}
    \mathrm{mmse}(Y \vert X) \coloneqq \inf_{h \textnormal{ meas.}} \mathbb{E}\left[(Y-h(X))^{2}\right],
\end{equation}
where the infimum is taken over all (Borel) measurable functions $h:\mathbb{R}^{d}\to\mathbb{R}$. The infimum in \eqref{eq:DefMMSE} is attained by the conditional expectation of $Y$ given $X$, i.e.,
\begin{equation}
\label{eq:MMSEMinimizer}
    \mathrm{mmse}(Y \vert X) = \mathbb{E}\left[(Y-\eta(X))^{2}\right],
\end{equation}
where $\eta(X) \stackrel{\textnormal{a.s.}}{=} \mathbb{E}\left[Y \vert X\right]$. Note that if $Y \stackrel{\textnormal{a.s.}}{=} h_{0}(X)$ for some function $h_{0}:\mathbb{R}^{d}\to\mathbb{R}$, then $\mathrm{mmse}(Y \vert X) = 0$. Also, note that if $X$ and $Y$ are independent, then the MMSE is maximal and $\mathrm{mmse}(Y \vert X) = \mathbb{E}\left[(Y-\mathbb{E}[Y])^{2}\right]$.

In the context of estimation-theoretic privacy, Asoodeh \textit{et al.} \cite{asoodeh2016privacy} introduced the notion of \emph{$\epsilon$-weak estimation privacy} to denote that
\begin{equation}
\label{eq:DefWeakEstimationPrivacy}
    \textnormal{mmse}(Y \vert X) \geq (1-\epsilon) \mathbb{E}\left[(Y-\mathbb{E}[Y])^{2}\right].
\end{equation}
Observe that, as defined in \eqref{eq:DefMMSE}, $\textnormal{mmse}(Y \vert X)$ quantifies the ability of an adversary to approximate $Y$, in the expected square-loss sense, upon observing $X$. Since a larger MMSE amounts to better privacy, estimation-theoretic privacy guarantees naturally come as lower bounds for $\textnormal{mmse}(Y \vert X)$ as expressed in \eqref{eq:DefWeakEstimationPrivacy}.

Also, when $Y$ is binary, the MMSE serves as a lower bound for the probability of error. Specifically, if $Y\in\{\pm1\}$, then
\begin{align}
    P_{\textnormal{error}}(Y \vert X) &= \inf_{h:\mathbb{R}^{d}\to\{\pm1\}} \mathbb{E}\left[\mathbbm{1}_{Y \neq h(X)}\right]\\
    &= \inf_{h:\mathbb{R}^{d}\to\{\pm1\}} \mathbb{E}\left[\frac{(Y-h(X))^{2}}{4}\right]\\
    &\geq \frac{1}{4} \inf_{h \textnormal{ meas.}} \mathbb{E}\left[(Y-h(X))^{2}\right]\\
    \label{eq:PerrorMMSE} &= \frac{1}{4} \mathrm{mmse}(Y \vert X).
\end{align}
Thus, for binary $Y$, any lower bound for $\mathrm{mmse}(Y \vert X)$ gives rise to a lower bound for $P_{\textnormal{error}}(Y \vert X)$. This observation further illustrates the importance of studying lower bounds for the MMSE in the context of privacy, where probability of correctly guessing ($1-P_{\textnormal{error}}$) has also been used as an information leakage measure \cite{braun2009quantitative,asoodeh2018estimation,diaz2019robustness}.

\subsection{Neural Network-based MMSE Estimation}
\label{Sec:NNMMSEEstimation}

A sigmoidal function $\phi:\mathbb{R}\to[-1,1]$ is a (measurable) function such that
\begin{equation}
\label{eq:DefSigmoidalFunction}
    \lim_{z\to-\infty} \phi(z) = -1 \quad \quad \text{and} \quad \quad \lim_{z\to\infty} \phi(z) = 1.
\end{equation}
Note that we are assuming that $\rvert\phi(z)\rvert\leq1$ for all $z\in\mathbb{R}$. Let $\mathcal{H}^{\phi}_{k}$ be the hypothesis class associated with a two-layer neural network of size $k$ with activation function $\phi$. More specifically, $\mathcal{H}^{\phi}_{k}$ is the set of all functions $h:\mathbb{R}^{d}\to\mathbb{R}$ of the form
\begin{equation}
\label{eq:DefHk}
    h(x) = c_{0} + \sum_{l=1}^{k} c_{l} \phi(a_{l} \cdot x + b_{l}),
\end{equation}
where $a_{l}\in\mathbb{R}^{d}$ and $b_{l},c_{l}\in\mathbb{R}$. In this work we propose the following neural network estimator of the MMSE of $Y$ given $X$. Given a random sample $\{(X_{i},Y_{i})\}_{i=1}^{n}$, we define
\begin{equation}
\label{eq:DefMMSEEstimator2}
    \mathrm{mmse}_{k,n}(Y \vert X) \coloneqq \inf_{h\in\mathcal{H}^{\phi}_{k}} \frac{1}{n} \sum_{i=1}^{n} (Y_{i} - h(X_{i}))^{2},
\end{equation}
i.e., $\mathrm{mmse}_{k,n}(Y \vert X)$ is the minimum empirical square-loss attained by a two-layer neural network. Observe that, optimization matters aside, $\mathrm{mmse}_{k,n}(Y \vert X)$ can be obtained from the sample using a device capable of implementing a two-layer neural network of size $k$. In this paper we take an information-theoretic perspective and assume infinite computational power. Specifically, we assume that $\mathrm{mmse}_{k,n}(Y \vert X)$ can be computed exactly\footnote{From an applied perspective, this assumption is not trivial to guarantee. Indeed, it is known that training neural networks to optimality could be a computationally difficult problem, see, e.g., \cite{blum1992training,bartlett1999hardness}. We further discuss some computational aspects in Section~\ref{Section:NumericalConsiderations}.}.

Our goal is to establish a (probabilistic) bound of the form
\begin{equation}
\label{eq:MainResultPrototype}
    \mathrm{mmse}_{k,n}(Y \vert X) - \epsilon_{k,n} \leq \mathrm{mmse}(Y \vert X),
\end{equation}
where $\epsilon_{k,n}$ is a positive number depending on the sample size $n$ and the neural network size $k$. In Section~\ref{Section:NNEstimation} we establish such a bound and, in addition, we derive an analogous result when the neural network has hyperbolic tangent as output activation function. Specifically, we replace $\mathcal{H}^{\phi}_{k}$ by $\tanh\circ\mathcal{H}^{\phi}_{k}$, i.e., the family of functions of the form $\tanh\circ h$ with $h\in\mathcal{H}^{\phi}_{k,n}$. In this case, the relevant MMSE estimator is the defined as
\begin{equation}
\label{eq:DefMMSEEstimator3}
    \mathrm{mmse}^{\ast}_{k,n}(Y \vert X) \coloneqq \inf_{h\in\mathcal{H}^{\phi}_{k}} \frac{1}{n} \sum_{i=1}^{n} (Y_{i} - \tanh(h(X_{i})))^{2}.
\end{equation}

Observe that, by definition, $\mathrm{mmse}(Y \vert X)$ is the minimum expected square-loss attained by \emph{any} measurable function. Hence, the bound in \eqref{eq:MainResultPrototype} differs from classical statistical learning results (e.g., Rademacher complexity bounds) for which the expected loss is minimized over the hypothesis class $\mathcal{H}^{\phi}_{k}$. In particular, we have to consider the so-called approximation error, which could be estimated via the function approximation theorem of Barron \cite{barron1993universal}.

\subsection{Barron's Theorem}

Let $B\subset\mathbb{R}^{d}$ be a bounded set such that $0\in B$. We define $\Gamma_{B}$ as the set of all functions $h:B\to\mathbb{R}$ admitting an integral representation of the form
\begin{equation}
\label{eq:DefGammaB}
    h(x) = h(0) + \int_{\mathbb{R}^{d}} \left(e^{\mathrm{i}\omega \cdot x} - 1\right) \hat{H}(\mathrm{d}\omega),
\end{equation}
for some complex-valued measure $\hat{H}$ such that $\int \lvert\omega\rvert \lvert\hat{H}\rvert(\mathrm{d}\omega)$ is finite. Observe that, as pointed out by Barron \cite[Sec.~III]{barron1993universal}, the right hand side of \eqref{eq:DefGammaB} defines an extension of $h$ to $\mathbb{R}^{d}$. However, it is important to remark that such an extension might not be unique as there might be multiple complex-valued measures $\hat{H}$ satisfying \eqref{eq:DefGammaB}.

Given $h\in\Gamma_{B}$, its Barron constant $C_{h}$ is defined as
\begin{equation}
\label{eq:DefBarronConstant}
    C_{h} \coloneqq \inf_{\hat{H}} \int_{\mathbb{R}^{d}} \lvert\omega\rvert_{B} \lvert\hat{H}\rvert(\mathrm{d}\omega),
\end{equation}
where the infimum is over all complex-valued measures $\hat{H}$ satisfying \eqref{eq:DefGammaB} and 
\begin{equation}
    \lvert\omega\rvert_{B} \coloneqq \sup_{x\in B} \lvert\omega \cdot x\rvert.
\end{equation}
To the best of the authors' knowledge, there is no known method to compute $C_{h}$ given an arbitrary $h\in\Gamma_{B}$. However, in practice, we can take any complex-valued measure $\hat{H}$ satisfying \eqref{eq:DefGammaB} and use it to evaluate the bound
\begin{equation}
\label{eq:BarronExtension1}
    C_{h} \leq \mathrm{rad}(B) \int_{\mathbb{R}^{d}} \lvert\omega\rvert \lvert\hat{H}\rvert(\mathrm{d}\omega),
\end{equation}
where $\displaystyle \mathrm{rad}(B) \coloneqq \sup_{x\in B} \lvert x \rvert$.

Under mild assumptions, the Barron constant could be related to the Fourier transform. Recall that, for a function $h\in L^{1}(\mathbb{R}^{d})$, its Fourier transform $\hat{h}:\mathbb{R}^{d}\to\mathbb{C}$ is defined as
\begin{equation}
    \hat{h}(\omega) \coloneqq \frac{1}{(2\pi)^{d/2}} \int_{\mathbb{R}^{d}} h(x) e^{-\mathrm{i}\omega \cdot x} \mathrm{d}x.
\end{equation}
If, in addition, $\hat{h}\in L^{1}(\mathbb{R}^{d})$ and $\int \lvert\omega\rvert \lvert\hat{h}(\omega)\rvert \mathrm{d}\omega$ is finite, then the Fourier inversion theorem implies that $h\vert_{B} \in \Gamma_{B}$ and
\begin{equation}
    C_{h\vert_{B}} \leq \frac{\mathrm{rad}(B)}{(2\pi)^{d/2}} \int_{\mathbb{R}^{d}} \lvert\omega\rvert \lvert\hat{h}(\omega)\rvert \mathrm{d}\omega.
\end{equation}

The following proposition establishes, in a quantitative manner, the universal approximating capabilities of two-layer neural networks. Observe that the statement below is a translation of Barron's original formulation \cite[Theorem~1]{barron1993universal} to the case of sigmoidal functions as defined in Section~\ref{Sec:NNMMSEEstimation}.\footnote{Observe that a sigmoidal function $\phi:\mathbb{R}\to[0,1]$ in the sense of Barron \cite[Sec.~I]{barron1993universal} can be converted into a sigmoidal function in the sense of Section~\ref{Sec:NNMMSEEstimation} by means of the transformation $\phi \mapsto 2(\phi-1/2)$.}

\begin{proposition}[Theorem~1, \cite{barron1993universal}]
\label{Prop:Barron}
Let $B\subset\mathbb{R}^{d}$ be a bounded set containing $0$. For every $h\in\Gamma_{B}$ and every probability distribution $P$ over $B$, there exists $h_{k}\in\mathcal{H}_{k}^{\phi}$ such that
\begin{equation}
    \int_{B} \lvert h_{k}(x) - h(x) \rvert^{2} P(\mathrm{d}x) \leq \frac{(2C_{h})^{2}}{k}.
\end{equation}
Furthermore, the coefficients of $h_{k}$ may be restricted to satisfy $|c_0| \leq |h(0)| + C_{h}$ and $\sum_{l=1}^{k} |c_{l}| \leq C_{h}$.
\end{proposition}

\subsection{Additive Gaussian Mechanism}
\label{Section:PrivacyPreservingMechanisms}

To motivate the forthcoming applications, consider the following setting. Assume that $X\in\mathbb{R}^{d}$ are sensible features of an individual which are correlated with a private attribute $Y\in\{\pm1\}$, e.g., $X$ could be height and $Y$ gender. Due to privacy concerns, a data analyst might not be able to observe $X$ but a \emph{sanitized} version of it. In this work we focus on the so-called additive Gaussian mechanism, a popular sanitization method in the information-theoretic and the differential privacy literature, see, e.g., \cite{dwork2014algorithmic,asoodeh2016privacy,asoodeh2016information}. Given $\sigma>0$, we define
\begin{equation}
    X_{i}^{\sigma} \coloneqq X_{i} + \sigma Z_{i},
\end{equation}
where $Z_{1},\ldots,Z_{n}$ are i.i.d.\ standard Gaussian vectors. Since the Gaussian distribution has unbounded support, it is often convenient to further process \emph{extreme values} of the random variables $X_{i}^{\sigma}$. We consider two processing techniques.

\subsubsection{Extreme Values Truncation}
\label{Subsection:ExtremeValuesTruncation}

Let $B\subset\mathbb{R}^{d}$ be a bounded set. Extreme values truncation is the data processing technique that discards all samples with $X_{i}^{\sigma}$ outside the set\footnote{Observe that this technique potentially reduces the sample size, although, the reduction is negligible when $B$ is large. In any case, for ease of notation, we let $n$ denote the effective sample size after truncation.} $B$. Let $\tilde{f}_{\pm}^{\sigma}$ be the conditional density of $X^{\sigma}$ after truncation given $Y=\pm1$. It is straightforward to verify that
\begin{equation}
\label{eq:DensityTruncation}
    \tilde{f}_{\pm}^{\sigma}(x) = \frac{(f_{\pm} \ast K_{\sigma})(x)}{\mathbb{P}\left(X^{\sigma} \in B \vert Y=\pm1 \right)} \mathbbm{1}_{x\in B},
\end{equation}
where $f_{\pm}$ is conditional density of $X$ given $Y = \pm 1$, $\ast$ is the convolution operator, and $K_{\sigma}$ is the density of $\sigma Z$.

\subsubsection{Extreme Values Randomization}
\label{Subsection:ExtremeValuesRandomization}

Let $B\subset\mathbb{R}^{d}$ be a bounded set. Extreme values randomization is the data processing technique that takes each $X_{i}^{\sigma}$ outside the set $B$ and replaces it with a random value on $B$. As before, let $\tilde{f}_{\pm}^{\sigma}$ be the conditional density of $X^{\sigma}$ after randomization given $Y=\pm1$. It is straightforward to verify that
\begin{equation}
\label{eq:DensityRandomization}
    \tilde{f}_{\pm}^{\sigma}(x) = \left[(f_{\pm} \ast K_{\sigma})(x) + \frac{\mathbb{P}\left(X^{\sigma} \not\in B \vert Y=\pm1 \right)}{\mathrm{vol}(B)}\right] \mathbbm{1}_{x\in B},
\end{equation}
where $\mathrm{vol}(B)$ denotes the volume of $B$ w.r.t.\ the Lebesgue measure on $\mathbb{R}^{d}$.

Under mild assumptions on $f_{\pm}$, e.g., bounded and compactly supported, both \eqref{eq:DensityTruncation} and \eqref{eq:DensityRandomization} define non-negative\footnote{In fact, $\tilde{f}_{\pm}^{\sigma}$ are bounded away from 0 on $B$.} smooth functions on the interior of $B$. If $B$ is closed, a routine application of Whitney's extension theorem \cite{whitney1934analytic} shows that, for any smooth function $h:\mathbb{R}\to\mathbb{R}$, the function $h\circ(\tilde{f}_{+}^{\sigma}/\tilde{f}_{-}^{\sigma})$ can be extended to a rapidly-decreasing smooth function over $\mathbb{R}^{d}$ \cite[Ch.~7]{rudin1991functional} and, in particular, $h\circ(\tilde{f}_{+}^{\sigma}/\tilde{f}_{-}^{\sigma})$ belongs to $\Gamma_{B}$.

\section{MMSE Lower Bounds}
\label{Section:NNEstimation}

In this section we provide lower bounds for $\mathrm{mmse}(Y \vert X)$ based on $\mathrm{mmse}_{k,n}(Y \vert X)$ and $\mathrm{mmse}^{\ast}_{k,n}(Y \vert X)$, as envisioned in \eqref{eq:MainResultPrototype}.

\subsection{Output Activation Function: Identity}
\label{Section:Estimation2LNN}

The following theorem establishes a lower bound for the MMSE in estimating $Y$ given $X$ based on the estimator $\mathrm{mmse}_{k,n}(Y \vert X)$, as defined in \eqref{eq:DefMMSEEstimator2}, and the Barron constant of the conditional expectation of $Y$ given $X$.

\begin{theorem}
\label{Theorem:2LNNIdentity}
Let $k,n\in\mathbb{N}$ and $B\subset\mathbb{R}^{d}$ be a bounded set containing $0$. If $Y\in[-1,1]$, $X$ is supported on $B$, and the conditional expectation $\eta(x) \coloneqq \mathbb{E}[Y \vert X = x]$ belongs to $\Gamma_{B}$, then, with probability at least $1-\delta$,
\begin{equation}
\label{eq:2LNNMainResultRepresentation}
    \mathrm{mmse}_{k,n}(Y \vert X) - \epsilon_{k,n,\delta} \leq \mathrm{mmse}(Y \vert X),
\end{equation}
where
\begin{equation}
\label{eq:2LNNMainResultEpsilon}
    \epsilon_{k,n,\delta} = 2(1+C_{\eta})^{2} \sqrt{\frac{2\log(1/\delta)}{n}} + \frac{4C_{\eta}^{2}}{k} + \frac{8C_{\eta}}{\sqrt{k}}.
\end{equation}
\end{theorem}

\begin{proof}
For ease of notation, we define
\begin{equation}
\label{eq:Proof2LNNDelta}
    \Delta \coloneqq \mathrm{mmse}_{k,n}(Y \vert X) - \mathrm{mmse}(Y \vert X).
\end{equation}
Also, we define $L(h) \coloneqq \mathbb{E}[(Y-h(X))^{2}]$ and
\begin{equation}
    \hat{L}_{n}(h) \coloneqq \frac{1}{n} \sum_{i=1}^{n} (Y_{i} - h(X_{i}))^{2}.
\end{equation}
Recall that the infimum defining $\mathrm{mmse}(Y \vert X)$ is attained by the conditional expectation $\eta$, see \eqref{eq:MMSEMinimizer}. Thus, we have that
\begin{align}
    \Delta &= \inf_{h\in\mathcal{H}^{\phi}_{k}} \hat{L}_{n}(h) - \inf_{h \textnormal{ meas.}} L(h)\\
    \label{eq:Proof2LNNDelta2} &= \inf_{h\in\mathcal{H}^{\phi}_{k}} \hat{L}_{n}(h) - L(\eta).
\end{align}

Since $\eta\in\Gamma_{B}$ by assumption, Barron's theorem (Proposition~\ref{Prop:Barron}) implies that there exists $\eta_k\in\mathcal{H}^{\phi}_k$ such that
\begin{equation}
\label{eq:ProofRepresentationBarron}
    \lVert \eta_k - \eta \rVert_{2} \leq \frac{2C_{\eta}}{\sqrt{k}},
\end{equation}
where $\lVert \cdot \rVert_{2}$ is the $2$-norm w.r.t.\ the distribution of $X$, i.e.,
\begin{equation}
    \lVert h \rVert_{2}^{2} = \int_{B} \lvert h(x) \rvert^{2} P_{X}(\mathrm{d}x).
\end{equation}
Furthermore, if we let
\begin{equation}
\label{eq:Proof2LNNetak}
    \eta_{k}(x) = c_{0} + \sum_{l=1}^{k} c_{l} \phi(a_{l} \cdot x + b_{l}),
\end{equation}
the coefficients $c_{0},c_{1},\ldots,c_{k}$ can be restricted to satisfy that $c_{0} \leq |\eta(0)| + C_{\eta}$ and $\sum_{l=1}^{k} |c_l| \leq C_{\eta}$. Observe that, by \eqref{eq:Proof2LNNDelta2},
\begin{equation}
\label{inq:RepresentationDeltaDecomposition}
    \Delta \leq \hat{L}_{n}(\eta_{k}) - L(\eta_{k}) + L(\eta_{k}) - L(\eta).
\end{equation}

Since $Y\in[-1,1]$, we have that $\lvert\eta(x)\rvert \leq 1$ for all $x\in B$. Recall that, by assumption, $\phi(z)\in[-1,1]$ for all $z\in\mathbb{R}$. Thus, by our choice of the coefficients $c_{0},c_{1},\ldots,c_{k}$ in \eqref{eq:Proof2LNNetak},
\begin{equation}
\label{eq:BoundInfNormEtak}
    \lvert\eta_{k}(x)\rvert \leq 1 + 2 C_{\eta}, \quad \quad x\in\mathbb{R}^{d}.
\end{equation}
Therefore, $(Y_{i}-\eta_{k}(X_{i}))^2 \leq 4(1+C_{\eta})^{2}$ for all $i\in\{1,\ldots,n\}$. As a result, a routine application of Hoeffding's inequality \cite[Sec.~4.2]{shalev2014understanding} implies that, with probability at least $1-\delta$,
\begin{equation}
\label{inq:RepresentationGeneralization}
    \hat{L}_{n}(\eta_{k}) - L(\eta_{k}) \leq 2(1+C_{\eta})^2 \sqrt{\frac{2\log(1/\delta)}{n}}.
 \end{equation}

It is straightforward to verify that
\begin{align}
    L(\eta_{k}) - L(\eta) &= \lVert\eta_{k} - \eta\rVert_{2}^{2} + 2\mathbb{E}[(\eta(X)-\eta_{k}(X))(Y-\eta(X)].
\end{align}
Recall that $\lvert\eta(x)\rvert \leq 1$ for all $x\in B$. Therefore, an application of the Cauchy–Schwarz inequality implies that
\begin{equation}
\label{eq:LipschitzcalL}
    \lvert L(\eta_{k}) - L(\eta) \rvert \leq \lVert\eta_{k} - \eta\rVert_{2}(4+\lVert\eta_{k} - \eta\rVert_{2}).
\end{equation}
By plugging \eqref{eq:ProofRepresentationBarron} in \eqref{eq:LipschitzcalL}, we conclude that
\begin{equation}
\label{inq:RepresentationApproximation}
    |L(\eta_k) - L(\eta)| \leq \frac{2C_{\eta}}{\sqrt{k}} \left(4 + \frac{2C_{\eta}}{\sqrt{k}}\right).
\end{equation}
The theorem follows by plugging \eqref{inq:RepresentationGeneralization} and \eqref{inq:RepresentationApproximation} in \eqref{inq:RepresentationDeltaDecomposition}.
\end{proof}

Regarding the assumptions of the previous theorem, it is important to remark that, in general, it might be non-trivial to verify that the conditional expectation $\eta$ belongs to $\Gamma_{B}$. Nonetheless, as shown in Section~\ref{Section:BarronConstantPostProcessing}, this assumption is automatically satisfied when data is post-processed by the additive Gaussian mechanism.

Note that $\epsilon_{k,n,\delta}$, as defined in \eqref{eq:2LNNMainResultEpsilon}, is non-increasing in $k$. Since $\mathcal{H}^{\phi}_{k}\subseteq\mathcal{H}^{\phi}_{k+1}$ for all $k\in\mathbb{N}$, $\mathrm{mmse}_{k,n}(Y \vert X)$ is also non-increasing in $k$. As a result, the lower bound in \eqref{eq:2LNNMainResultRepresentation} does not necessarily improve by making $k$ larger. Indeed, it is known that if $k$ is large enough then $\mathrm{mmse}_{k,n}(Y \vert X)$ is equal to 0, see, e.g., \cite{baum1988capabilities,bubeck2020network}, which makes the lower bound in \eqref{eq:2LNNMainResultRepresentation} trivial. Together with the fact that the minimization defining $\mathrm{mmse}_{k,n}(Y \vert X)$ becomes harder as $k$ increases, the previous observations reveal the non-trivial nature of finding the value of $k$ that produces the best numerical results. We expand on this discussion in Section~\ref{Section:NumericalConsiderations}.

\subsection{Output Activation Function: Hyperbolic Tangent}

The following theorem establishes a lower bound for the MMSE in estimating $Y$ given $X$ based on the estimator $\mathrm{mmse}^{\ast}_{k,n}(Y \vert X)$, as defined in \eqref{eq:DefMMSEEstimator3}, and the Barron constant of the log-likelihood ratio defined in \eqref{eq:3LNNDefTheta} below. Note that, unlike Theorem~\ref{Theorem:2LNNIdentity}, the next theorem requires $Y$ to be binary.

\begin{theorem}
\label{Theorem:2LNNtanh}
Let $k,n\in\mathbb{N}$ and $B\subset\mathbb{R}^{d}$ be a bounded set containing $0$. Assume that $Y\in\{\pm1\}$, $X$ is supported on $B$, and the conditional density of $X$ given $Y=\pm1$, denoted by $f_{\pm}$, is positive on $B$. Let $p \coloneqq \mathbb{P}(Y=1)$ and
\begin{equation}
\label{eq:3LNNDefTheta}
    \theta(x) \coloneqq \frac{1}{2} \log\left(\frac{pf_{+}(x)}{\bar{p}f_{-}(x)}\right), \quad \quad x\in B.
\end{equation}
If $\theta$ belongs to $\Gamma_{B}$, then, with probability at least $1-\delta$,
\begin{equation}
\label{eq:3LNNMainResultRepresentation}
    \mathrm{mmse}^{\ast}_{k,n}(Y \vert X) - \epsilon^{\ast}_{k,n,\delta} \leq \mathrm{mmse}(Y \vert X),
\end{equation}
where
\begin{equation}
    \epsilon^{\ast}_{k,n,\delta} = 2\sqrt{\frac{2\log(1/\delta)}{n}} + \frac{4C_{\theta}^{2}}{k} + \frac{8C_{\theta}}{\sqrt{k}}.
\end{equation}
\end{theorem}

\begin{proof}
For ease of notation, we define
\begin{equation}
\label{eq:Proof3LNNDelta}
    \Delta^{\ast} \coloneqq \mathrm{mmse}^{\ast}_{k,n}(Y \vert X) - \mathrm{mmse}(Y \vert X).
\end{equation}
Also, we define $L^{\ast}(h) \coloneqq \mathbb{E}[(Y-\tanh(h(X)))^{2}]$ and
\begin{equation}
    \hat{L}^{\ast}_{n}(h) \coloneqq \frac{1}{n} \sum_{i=1}^{n} (Y_{i} - \tanh(h(X_{i})))^{2}.
\end{equation}
Recall that the infimum defining $\mathrm{mmse}(Y \vert X)$ is attained by the conditional expectation $\eta$, see \eqref{eq:MMSEMinimizer}. A straightforward computation shows that, for every $x\in B$,
\begin{align}
    \eta(x) &= \frac{pf_{+}(x) - \bar{p}f_{-}(x)}{pf_{+}(x) + \bar{p}f_{-}(x)}\\
    &= \tanh\left(\frac{1}{2}\log\left(\frac{pf_{+}(x)}{\bar{p}f_{-}(x)}\right)\right)\\
    \label{eq:EtaThetaRelation} &= \tanh\left(\theta(x)\right).
\end{align}
Therefore, we have that
\begin{align}
    \mathrm{mmse}(Y \vert X) &= \mathbb{E}[(Y-\eta(X))^{2}]\\
    &= \mathbb{E}[(Y-\tanh(\theta(X)))^{2}]\\
    &= L^{\ast}(\theta),
\end{align}
and, as a result,
\begin{equation}
\label{eq:Proof3LNNDelta2}
    \Delta^{\ast} = \inf_{h\in\mathcal{H}^{\phi}_{k}} \hat{L}^{\ast}_{n}(h) - L^{\ast}(\theta).
\end{equation}

Since $\theta\in\Gamma_{B}$ by assumption, Barron's theorem (Proposition~\ref{Prop:Barron}) implies that there exists $\theta_k\in\mathcal{H}^{\phi}_k$ such that
\begin{equation}
\label{eq:Proof3LNNBarron}
    \lVert \theta_k - \theta \rVert_{2} \leq \frac{2C_{\theta}}{\sqrt{k}},
\end{equation}
where $\lVert \cdot \rVert_{2}$ is the $2$-norm w.r.t.\ the distribution of $X$. Furthermore, if we let
\begin{equation}
\label{eq:etak}
    \theta_{k}(x) = c_{0} + \sum_{l=1}^{k} c_{l} \phi(a_{l} \cdot x + b_{l}),
\end{equation}
the coefficients $c_{0},c_{1},\ldots,c_{k}$ can be restricted to satisfy that $c_{0} \leq |\theta(0)| + C_{\theta}$ and $\sum_{l=1}^{k} |c_l| \leq C_{\theta}$. Observe that, by \eqref{eq:Proof3LNNDelta2},
\begin{equation}
\label{inq:Proof3LNNDelta}
    \Delta \leq \hat{L}^{\ast}_{n}(\theta_{k}) - L^{\ast}(\theta_{k}) + L^{\ast}(\theta_{k}) - L^{\ast}(\theta).
\end{equation}

Since $Y\in\{\pm1\}$ and $\tanh:\mathbb{R}\to(-1,1)$, it is immediate to verify that $(Y_{i}-\tanh(\theta_{k}(X_{i})))^2 \leq 4$ for all $i\in\{1,\ldots,n\}$. As a result, a routine application of Hoeffding's inequality \cite[Sec.~4.2]{shalev2014understanding} implies that, with probability at least $1-\delta$,
\begin{equation}
\label{inq:Proof3LNNLargeDeviations}
    \hat{L}^{\ast}_{n}(\theta_{k}) - L^{\ast}(\theta_{k}) \leq 2\sqrt{\frac{2\log(1/\delta)}{n}}.
 \end{equation}

It is straightforward to verify that
\begin{align}
    L^{\ast}(\theta_{k}) - L^{\ast}(\theta) =& \lVert\tanh\circ\theta_{k} - \tanh\circ\theta\rVert_{2}^{2} + 2\mathbb{E}[(\tanh(\theta(X))-\tanh(\theta_{k}(X)))(Y-\tanh(\theta(X)))].
\end{align}
Thus, the Cauchy–Schwarz inequality leads to
\begin{align}
    \lvert L^{\ast}(\theta_{k}) - L^{\ast}(\theta) \rvert &\leq \lVert\tanh\circ\theta_{k} - \tanh\circ\theta\rVert_{2}(4+\lVert\tanh\circ\theta_{k} - \tanh\circ\theta\rVert_{2}).
\end{align}
Since $\tanh$ is a 1-Lipschitz function, it can be verified that $\lVert\tanh\circ\theta_{k} - \tanh\circ\theta\rVert_{2} \leq \lVert\theta_{k} - \theta\rVert_{2}$. Therefore, we have that
\begin{equation}
\label{eq:Proof3LNNLipschitz}
    \lvert L^{\ast}(\theta_{k}) - L^{\ast}(\theta) \rvert \leq \lVert\theta_{k} - \theta\rVert_{2} (4+\lVert\theta_{k} - \theta\rVert_{2}).
\end{equation}
By plugging \eqref{eq:Proof3LNNBarron} in \eqref{eq:Proof3LNNLipschitz}, we conclude that
\begin{equation}
\label{inq:Proof3LNNFunctionApproximation}
    \lvert L^{\ast}(\theta_{k}) - L^{\ast}(\theta) \rvert \leq \frac{2C_{\theta}}{\sqrt{k}} \left(4 + \frac{2C_{\theta}}{\sqrt{k}}\right).
\end{equation}
The theorem follows by plugging \eqref{inq:Proof3LNNLargeDeviations} and \eqref{inq:Proof3LNNFunctionApproximation} in \eqref{inq:Proof3LNNDelta}.
\end{proof}

As with Theorem~\ref{Theorem:2LNNIdentity}, the hypotheses of the previous theorem are automatically satisfied when data is post-processed by the additive Gaussian mechanism. We discuss this claim in detail in Section~\ref{Section:BarronConstantPostProcessing}.

Observe that, as established in \eqref{eq:EtaThetaRelation}, the conditional expectation $\eta$ and the log-likelihood ratio $\theta$ satisfy that
\begin{equation}
    \eta = \tanh \circ \theta.
\end{equation}
In view of this relation, the choice of the neural network defining $\mathrm{mmse}^{\ast}_{k,n}(Y \vert X)$ becomes evident: the second layer approximates $\theta$ while the output activation function is the hyperbolic tangent function.

\section{A General Bound for the Barron Constant}
\label{Section:GeneralBoundBarronConstant}

Theorem~\ref{Theorem:2LNNIdentity} establishes a lower bound for $\mathrm{mmse}(Y \vert X)$ based on the estimator $\mathrm{mmse}_{k,n}(Y \vert X)$ and the Barron constant $C_{\eta}$ of the conditional expectation of $Y$ given $X$. While $\mathrm{mmse}_{k,n}(Y \vert X)$ can be computed from the sample, providing estimates for the Barron constant $C_{\eta}$ might be challenging for two reasons: (i) the conditional expectation of $Y$ given $X$ depends on the distribution of $X$ and $Y$, which is typically unavailable in practice, and (ii) the Barron constant $C_{\eta}$ is defined in terms of the Fourier transform of $\eta$, which makes its computation unfeasible in most cases. (A similar remark applies, \emph{mutatis mutandis}, to Theorem~\ref{Theorem:2LNNtanh}.) In this section we provide some results that alleviate the second issue; the discussion of the first issue is left for the following section.

\subsection{1-Dimensional Bound}
\label{Subsection:1DBound}

In this section we focus on a special family of real-valued functions of a real variable whose Barron's constant admits a relatively tractable representation.

Let $h:\mathbb{R}\to\mathbb{R}$ be a differentiable function. If $h'\in L^{1}(\mathbb{R})$ and $\widehat{h'}\in L^{1}(\mathbb{R})$, the Fourier inversion theorem implies that
\begin{equation}
\label{eq:BarronConstantDerivative}
    h'(x) = \frac{1}{\sqrt{2\pi}} \int_{\mathbb{R}} \widehat{h'}(\omega) e^{i\omega x} \mathrm{d}\omega.
\end{equation}
As pointed out by Barron \cite[Appendix]{barron1993universal}, \eqref{eq:BarronConstantDerivative} implies that $h\vert_{B}$ belongs to $\Gamma_{B}$ for every bounded set $B$ containing 0 and
\begin{equation}
\label{eq:BarronExtension2}
    C_{h\vert_{B}} \leq \frac{\mathrm{rad}(B)}{\sqrt{2\pi}} \int_{\mathbb{R}} \lvert\widehat{h'}(\omega)\rvert \mathrm{d}\omega.
\end{equation}
Thus, by abuse of notation, we define
\begin{equation}
\label{eq:DefChDifferentiable}
    C_{h} \coloneqq \frac{1}{\sqrt{2\pi}} \int_{\mathbb{R}} \lvert\widehat{h'}(\omega)\rvert \mathrm{d}\omega,
\end{equation}
whenever $h:\mathbb{R}\to\mathbb{R}$ satisfies that $h',\widehat{h'}\in L^{1}(\mathbb{R})$.

\begin{theorem}
\label{Theorem:CEtaL1NormsBound}
Let $h:\mathbb{R}\to\mathbb{R}$ be a thrice differentiable function. If $h',h'',h'''\in L^{1}(\mathbb{R})$ and vanish at infinity, then
\begin{equation}
    C_{h} \leq \frac{2\sqrt{2}}{\sqrt{\pi}} \left(1 + \log\left(\frac{\sqrt{\lVert h'\rVert_1 \lVert h'''\rVert_1}}{\lVert h''\rVert_{1}}\right)\right) \lVert h''\rVert_{1}.
\end{equation}
\end{theorem}

\begin{proof}
Let $0<\lambda_{1}<\lambda_{2}$. We split the integral in \eqref{eq:DefChDifferentiable} as
\begin{align}
    \mathrm{I} &\coloneqq \int_{-\lambda_{1}}^{\lambda_{1}} \lvert\widehat{h'}(\omega)\rvert \mathrm{d}\omega,\\
    \mathrm{II} &\coloneqq  \left(\int_{\lambda_{1}}^{\lambda_{2}} + \int_{-\lambda_{2}}^{-\lambda_{1}}\right) \lvert\widehat{h'}(\omega)\rvert \mathrm{d}\omega,\\
    \mathrm{III} &\coloneqq \left(\int_{\lambda_{2}}^{\infty} + \int_{-\infty}^{-\lambda_{2}}\right) \lvert\widehat{h'}(\omega)\rvert \mathrm{d}\omega.
\end{align}
First, observe that
\begin{equation}
\label{eq:CEtaLemmaBoundPartI}
    \mathrm{I} \leq 2 \lVert\widehat{h'}\rVert_\infty \lambda_{1} \leq 2 \lVert h' \rVert_1 \lambda_{1},
\end{equation}
where we applied the inequality $\lVert \widehat{h'} \rVert_{\infty} \leq \lVert h' \rVert_{1}$. Since $h'$ vanishes at infinity and $h''\in L^{1}(\mathbb{R})$, $\widehat{h''}(\omega) = \mathrm{i} \omega \widehat{h'}(\omega)$ for every $\omega\in\mathbb{R}$. Thus, we have that
\begin{align}
    \mathrm{II} &=  \left(\int_{\lambda_{1}}^{\lambda_{2}} + \int_{-\lambda_{2}}^{-\lambda_{1}}\right) \frac{1}{\lvert\omega\rvert} \lvert\widehat{h''}(\omega)\rvert \mathrm{d} \omega\\
    &\leq 2 \lVert\widehat{h''}\rVert_{\infty} \log\left(\frac{\lambda_{2}}{\lambda_{1}}\right)\\
    \label{eq:CEtaLemmaBoundPartII} &\leq 2 \lVert h''\rVert_{1} \log\left(\frac{\lambda_{2}}{\lambda_{1}}\right).
\end{align}
Similarly, $\widehat{h'''}(\omega) = (\mathrm{i} \omega)^{2} \widehat{h'}(\omega)$ for every $\omega\in\mathbb{R}$ and
\begin{equation}
\label{eq:CEtaLemmaBoundPartIII}
    {\rm III} = \left(\int_{\lambda_{2}}^{\infty} + \int_{-\infty}^{-\lambda_2}\right) \frac{1}{\omega^2} \big|\widehat{h'''}(\omega)\big| \mathrm{d} \omega \leq 2\frac{\lVert h'''\rVert_{1}}{\lambda_{2}}.
\end{equation}
By plugging \eqref{eq:CEtaLemmaBoundPartI}, \eqref{eq:CEtaLemmaBoundPartII} and \eqref{eq:CEtaLemmaBoundPartIII} in \eqref{eq:DefChDifferentiable}, we conclude that $\widehat{h'}\in L^{1}(\mathbb{R})$ and
\begin{equation}
\label{eq:CEtaLemmaFinalEq}
    C_{h} \leq \sqrt{\frac{2}{\pi}}\left(\lVert h'\rVert_{1} \lambda_{1} + \lVert h''\rVert_{1} \log\left(\frac{\lambda_{2}}{\lambda_{1}}\right) + \frac{\lVert h'''\rVert_{1}}{\lambda_{2}}\right).
\end{equation}
By taking $\lambda_{1} = \frac{\lVert h''\rVert_{1}}{\lVert h'\rVert_{1}}$ and $\lambda_{2} = \frac{\lVert h'''\rVert_{1}}{\lVert h''\rVert_{1}}$, the result follows.
\end{proof}

Observe that if we let $\lambda_{1} = \lambda_{2} = \sqrt{\lVert h'''\rVert_{1}/\lVert h' \rVert_{1}}$ in \eqref{eq:CEtaLemmaFinalEq}, we obtain
\begin{equation}
\label{eq:SimplerBoundCh}
    C_{h} \leq \frac{2\sqrt{2}}{\sqrt{\pi}} \sqrt{\lVert h'\rVert_1 \lVert h'''\rVert_1}.
\end{equation}
This bound is generalized for functions $h:\mathbb{R}^{d}\to\mathbb{R}$ in Theorem~\ref{Theorem:CEtaL1NormsBoundHighDimensions} below. Observe that while the bound in \eqref{eq:SimplerBoundCh} is simpler than the one provided in Theorem~\ref{Theorem:CEtaL1NormsBound}, it is typically weaker in applications. 

Since Theorem~\ref{Theorem:CEtaL1NormsBound} will be applied to the conditional expectation $\eta$ and the log-likelihood ratio $\theta$, we need to compute the derivatives of these functions. The following lemma provides useful expressions for the first three derivatives of $\eta$ in the case when $Y$ is binary. Recall that if $f_{\pm}:B\to\mathbb{R}$ is the conditional density of $X$ given $Y=\pm1$ and $p = \mathbb{P}(Y = 1)$, then the conditional expectation of $Y$ given $X$ is given by
\begin{equation}
\label{eq:DefEta}
    \eta(x) = \frac{pf_{+}(x) - \bar{p}f_{-}(x)}{pf_{+}(x) + \bar{p}f_{-}(x)}.
\end{equation}

\begin{lemma}
\label{Lemma:DerivativesKappa}
If $\eta$ is defined as in \eqref{eq:DefEta}, then
\begin{align}
    \label{eq:Kappa0} \eta' &= 2\frac{g_{+}' g_{-} - g_{+} g_{-}'}{(g_{+} + g_{-})^2},\\
    \label{eq:Kappa1} \eta'' &= 2 \frac{g_{+}''g_{-} - g_{+}g_{-}''}{(g_{+} + g_{-})^2} - 2 \eta' \frac{g_{+}' + g_{-}'}{g_{+} + g_{-}},\\
    \label{eq:Kappa2} \eta''' &= 2 \frac{g_{+}'''g_{-} + g_{+}''g_{-}' - g_{+}'g_{-}'' - g_{+}g_{-}'''}{(g_{+} + g_{-})^2} - 2 \eta' \frac{g_{+}'' + g_{-}''}{g_{+} + g_{-}} - 3 \eta'' \frac{g_{+}' + g_{-}'}{g_{+} + g_{-}},
\end{align}
where $g_{+} = pf_{+}$ and $g_{-} = \bar{p}f_{-}$.
\end{lemma}

\begin{proof}
Equation \eqref{eq:Kappa0} follows easily from the \eqref{eq:DefEta}. By the quotient rule $\displaystyle \left(\frac{h_{1}}{h_{2}}\right)' = \frac{h_{1}'}{h_{2}} - \frac{h_{1}}{h_{2}} \frac{h_{2}'}{h_{2}}$, \eqref{eq:Kappa1} follows from \eqref{eq:Kappa0}. Using similar arguments, \eqref{eq:Kappa2} follows from \eqref{eq:Kappa1}.
\end{proof}

Similarly, the following lemma provides useful expressions for the first three derivatives of the log-likelihood ratio
\begin{equation}
\label{eq:DefTheta}
    \theta(x) = \frac{1}{2} \log\left(\frac{pf_{+}(x)}{\bar{p}f_{-}(x)}\right).
\end{equation}

\begin{lemma}
\label{Lemma:DerivativesTheta}
If $\theta$ is defined as in \eqref{eq:DefTheta}, then
\begin{align}
    \label{eq:D1theta} 2\theta' &= \frac{g_{+}'}{g_{+}} - \frac{g_{-}'}{g_{-}},\\
    \label{eq:D2theta} 2\theta'' &= \left[\frac{g_{+}''}{g_{+}}-\left(\frac{g_{+}'}{g_{+}}\right)^{2}\right] - \left[\frac{g_{-}''}{g_{-}}-\left(\frac{g_{-}'}{g_{-}}\right)^{2}\right],\\
    \label{eq:D3theta} 2\theta''' &= \left[\frac{g_{+}'''}{g_{+}} - 3 \frac{g_{+}'g_{+}''}{g_{+}^{2}} + 2 \left(\frac{g_{+}'}{g_{+}}\right)^{3}\right] - \left[\frac{g_{-}'''}{g_{-}} - 3 \frac{g_{-}'g_{-}''}{g_{-}^{2}} + 2 \left(\frac{g_{-}'}{g_{-}}\right)^{3}\right],
\end{align}
where $g_{+} = pf_{+}$ and $g_{-} = \bar{p}f_{-}$.
\end{lemma}

\begin{proof}
The identities in \eqref{eq:D1theta} -- \eqref{eq:D3theta} follow from the quotient rule and the logarithmic derivative $(\log h)' = h'/h$.
\end{proof}

It is important to remark that, in general, $\eta$ and $\theta$ might not satisfy the assumptions of Theorem~\ref{Theorem:CEtaL1NormsBound}, i.e., having integrable derivatives that vanish at infinity. Nonetheless, as shown in Section~\ref{Section:BarronConstantPostProcessing}, they satisfy the aforementioned assumptions when data is post-processed by the additive Gaussian mechanism. We verify this claim through a careful analysis of the derivatives of $\eta$ and $\theta$ given in Lemmas~\ref{Lemma:DerivativesKappa} and \ref{Lemma:DerivativesTheta}, respectively.

\subsection{$d$-Dimensional Extension}

In this section we focus on a special family of real-valued functions of $d$-real variables whose Barron's constant admits a relatively tractable representation.

Let $h:\mathbb{R}^{d}\to\mathbb{R}$ be a differentiable function. For ease of notation, we let $h_{x_{j}} \coloneqq \frac{\partial}{\partial x_{j}} h$. If $h_{x_{j}},\widehat{h_{x_{j}}}\in L^{1}(\mathbb{R}^{d})$ for every $j\in\{1,\ldots,d\}$, the Fourier inversion theorem implies that
\begin{equation}
\label{eq:BarronConstantDerivativedDDimensional}
    \nabla h(x) = \frac{1}{(2\pi)^{d/2}} \int_{\mathbb{R}^{d}} \widehat{\nabla h}(\omega) e^{i\omega \cdot x} \mathrm{d}\omega,
\end{equation}
where $\nabla h = (h_{x_{1}},\ldots,h_{x_{d}})$ and $\widehat{\nabla h} = (\widehat{h_{x_{1}}},\ldots,\widehat{h_{x_{d}}})$. As pointed out by Barron \cite[Appendix]{barron1993universal}, \eqref{eq:BarronConstantDerivativedDDimensional} implies that $h\vert_{B}$ belongs to $\Gamma_{B}$ for every bounded set $B$ containing 0 and
\begin{equation}
    C_{h\vert_{B}} \leq \frac{\mathrm{rad}(B)}{(2\pi)^{d/2}} \int_{\mathbb{R}^{d}} \lvert\widehat{\nabla h}(\omega)\rvert \mathrm{d}\omega.
\end{equation}
Thus, by abuse of notation, we define
\begin{equation}
\label{eq:DefChDifferentiableDDimensional}
    C_{h} \coloneqq \frac{1}{(2\pi)^{d/2}} \int_{\mathbb{R}^{d}} \lvert\widehat{\nabla h}(\omega)\rvert \mathrm{d}\omega,
\end{equation}
whenever $h:\mathbb{R}^{d}\to\mathbb{R}$ satisfies that $h_{x_{j}},\widehat{h_{x_{j}}}\in L^{1}(\mathbb{R}^{d})$ for every $j\in\{1,\ldots,d\}$.

\begin{theorem}
\label{Theorem:CEtaL1NormsBoundHighDimensions}
Let $h:\mathbb{R}^{d}\to\mathbb{R}$ be a function of class $C^{d+2}$. If the partial derivatives of $h$ of order up to $d+2$ belong to $L^{1}(\mathbb{R}^{d})$ and vanish at infinity, then
\begin{equation}
\label{eq:ChBoundDDimensional}
    C_{h} \leq A_d N_{1}^{1/(d+1)} N_{2}^{d/(d+1)},
\end{equation}
where $\displaystyle A_{d} = \frac{d+1}{d^{d/(d+1)}} \frac{d^{d/2}}{2^{d/2}\Gamma(d/2+1)}$, $\displaystyle N_{1}^{2} = \sum_{j=1}^{d} \lVert h_{x_{j}} \rVert_{1}^{2}$ and
\begin{equation}
    N_{2}^{2} = \sum_{j=1}^{d} \left(\sum_{j'=1}^{d} \left\lVert\frac{\partial^{d+1}}{\partial x_{j'}^{d+1}} h_{x_{j}}\right\rVert_{1}\right)^{2}.
\end{equation}
\end{theorem}

\begin{proof}
For $\lambda>0$, let $B_{\lambda}$ be the $d$-dimensional ball of radius $\lambda$. We split the integral in \eqref{eq:DefChDifferentiableDDimensional} as
\begin{equation}
    \mathrm{I} \coloneqq \int_{B_{\lambda}} \lvert\widehat{\nabla h}(\omega)\rvert \mathrm{d}\omega \quad \text{and} \quad \mathrm{II} \coloneqq \int_{B_{\lambda}^{c}} \lvert\widehat{\nabla h}(\omega)\rvert \mathrm{d}\omega.
\end{equation}
Observe that, for every $\omega\in\mathbb{R}^{d}$,
\begin{equation}
\label{eq:ProofChHighDimensional}
    \lvert\widehat{\nabla h}(\omega)\rvert^{2} = \sum_{j=1}^{d} \lvert\widehat{h_{x_{j}}}(\omega)\rvert^{2} \leq \sum_{j=1}^{d} \lVert h_{x_{j}} \rVert_{1}^{2} \eqqcolon N_{1}^{2},
\end{equation}
where we applied the inequality $\lVert\widehat{h_{x_{j}}}\rVert_{\infty} \leq \lVert h_{x_{j}} \rVert_{1}$. Therefore,
\begin{equation}
\label{eq:CEtaLemmaBoundPartIHighDimensions}
    \mathrm{I} \leq N_{1} \mathrm{vol}(B_{\lambda}) = \frac{\pi^{d/2}N_{1}\lambda^{d}}{\Gamma(d/2+1)},
\end{equation}
as $\mathrm{vol}(B_{1}) = \pi^{d/2}/\Gamma(d/2+1)$.

For $\omega\in\mathbb{R}^{d}$, the generalized mean inequality asserts that
\begin{equation}
    \sqrt[2]{\frac{\lvert\omega_{1}\rvert^{2}+\cdots+\lvert\omega_{d}\rvert^{2}}{d}} \leq \sqrt[d+1]{\frac{\lvert\omega_{1}\rvert^{d+1}+\cdots+\lvert\omega_{d}\rvert^{d+1}}{d}},
\end{equation}
which in turn leads to
\begin{equation}
\label{eq:ChDDimensionalProofNormOmega}
    \lvert\omega\rvert^{d+1} \leq d^{(d-1)/2} \sum_{j'=1}^{d} \lvert\omega_{j'}\rvert^{d+1}.
\end{equation}
For ease of notation, we define $\partial_{j'}^{d+1} \coloneqq \frac{\partial^{d+1}}{\partial \omega_{j'}^{d+1}}$. Since the partial derivatives of $h$ of order up to $d+2$ belong to $L^{1}(\mathbb{R}^{d})$ and vanish at infinity, then, for every $\omega\in\mathbb{R}^{d}$,
\begin{equation}
\label{eq:ChDDimensionalProofFourier}
    \widehat{\partial_{j'}^{d+1}h_{x_{j}}}(\omega) = (\mathrm{i} \omega_{j'})^{d+1} \widehat{h_{x_{j}}}(\omega).
\end{equation}
Therefore, \eqref{eq:ChDDimensionalProofNormOmega} and \eqref{eq:ChDDimensionalProofFourier} imply that
\begin{align}
    \lvert\omega\rvert^{d+1} \lvert\widehat{h_{x_{j}}}(\omega)\rvert &\leq d^{(d-1)/2} \sum_{j'=1}^{d} \lvert\omega_{j'}\rvert^{d+1} \lvert\widehat{h_{x_{j}}}(\omega)\rvert\\
    &= d^{(d-1)/2} \sum_{j'=1}^{d} \lvert\widehat{\partial_{j'}^{d+1} h_{x_{j}}}(\omega)\rvert\\
    &\leq d^{(d-1)/2} \sum_{j'=1}^{d} \lVert\partial_{j'}^{d+1} h_{x_{j}}\rVert_{1},
\end{align}
where we applied the inequality $\lVert\widehat{\partial_{j'}^{d+1}h_{x_{j}}}\rVert_{\infty} \leq \lVert\partial_{j'}^{d+1} h_{x_{j}}\rVert_{1}$. Alternatively, we have that
\begin{equation}
    \lvert\widehat{h_{x_{j}}}(\omega)\rvert \leq \frac{d^{(d-1)/2}}{\lvert\omega\rvert^{d+1}} \sum_{j'=1}^{d} \lVert\partial_{j'}^{d+1} h_{x_{j}}\rVert_{1}.
\end{equation}
As a result, we obtain that
\begin{align}
    \lvert\widehat{\nabla h}(\omega)\rvert &= \left(\sum_{j=1}^{d} \lvert\widehat{h_{x_{j}}}(\omega)\rvert^{2}\right)^{1/2}\\
    &\leq \frac{d^{(d-1)/2}}{\lvert\omega\rvert^{d+1}} \left(\sum_{j=1}^{d} \left(\sum_{j'=1}^{d} \lVert\partial_{j'}^{d+1} h_{x_{j}}\rVert_{1}\right)^{2}\right)^{1/2}\\
    &= \frac{d^{(d-1)/2}}{\lvert\omega\rvert^{d+1}} N_{2}.
\end{align}
Since $\omega \mapsto 1/\lvert\omega\rvert^{d+1}$ is a radial function, we have that
\begin{align}
    \mathrm{II} &\leq d^{(d-1)/2} N_{2} \int_{B_{\lambda}^{c}} \frac{1}{\lvert\omega\rvert^{d+1}} \mathrm{d}\omega\\
    &= d^{(d-1)/2} N_{2} \int_{\lambda}^{\infty} \frac{1}{r^{d+1}} \frac{d\pi^{d/2}r^{d-1}}{\Gamma(d/2+1)} \mathrm{d}r\\
    &= \frac{\pi^{d/2}d^{(d+1)/2}N_{2}}{\Gamma(d/2+1)} \int_{\lambda}^{\infty} \frac{1}{r^{2}} \mathrm{d}r\\
    \label{eq:CEtaLemmaBoundPartIIHighDimensions} &= \frac{\pi^{d/2}d^{(d+1)/2}N_{2}}{\Gamma(d/2+1)\lambda}.
\end{align}
By plugging \eqref{eq:CEtaLemmaBoundPartIHighDimensions} and \eqref{eq:CEtaLemmaBoundPartIIHighDimensions} in \eqref{eq:DefChDifferentiableDDimensional}, we conclude that
\begin{equation}
\label{eq:CEtaLemmaFinalEqDDimensional}
    C_{h} \leq \frac{1}{2^{d/2}\Gamma(d/2+1)} \left(N_{1} \lambda^{d} + \frac{d^{\frac{d+1}{2}}N_{2}}{\lambda}\right).
\end{equation}
By taking $\displaystyle \lambda^{d+1} = \frac{d^{(d-1)/2}N_{2}}{N_{1}}$, the result follows.
\end{proof}

Note that \eqref{eq:ChBoundDDimensional} generalizes \eqref{eq:SimplerBoundCh}, as it only involves derivatives of order 1 and $d+2$. While it is possible to establish bounds that more closely resemble Theorem~\ref{Theorem:CEtaL1NormsBound} (e.g., by using derivatives of order 1, $d+1$ and $d+2$), they are more convoluted than \eqref{eq:ChBoundDDimensional} and add little practical value.

Theorem~\ref{Theorem:CEtaL1NormsBoundHighDimensions} might not be straightforward to apply as it heavily depends on the higher order partial derivatives of $h$. Furthermore, since $d^{1/d} \sim 1$ and $\Gamma(z+1) \sim \sqrt{2\pi z} \left(\frac{z}{e}\right)^{z}$, we can show that
\begin{equation}
    A_{d} \sim \sqrt{\frac{e^{d}}{\pi d}}.
\end{equation}
In particular, the bound in \eqref{eq:ChBoundDDimensional} has an exponential dependency on the dimension. Despite the negative nature of this observation, it is indeed natural in view of a similar comment made by Barron in \cite[Sec.~IX-9]{barron1993universal}.

\begin{remark}
A line of research initiated by Breiman \cite{breiman1993hinging}, and recently extended by Domingo-Enrich and Mroueh \cite{domingo2021tighter} building upon the results of Ongie \textit{et al.} \cite{ongie2019function}, focuses on approximation results for two-layer neural networks with ReLU activation functions. In the 1-dimensional case, the approximation results by Breiman \cite{breiman1993hinging} rely on a variation of the Barron constant given by
\begin{equation}
    C_{h}' \coloneqq \frac{1}{\sqrt{2\pi}} \int_{\mathbb{R}} \lvert w \rvert^{2} \lvert \hat{h}(\omega) \rvert \mathrm{d}\omega.
\end{equation}
It is important to remark that Theorems~\ref{Theorem:CEtaL1NormsBound} and \ref{Theorem:CEtaL1NormsBoundHighDimensions} can be generalized to $C_{h}'$ at the expense of increasing by one the order of the derivatives involved. This seemingly superficial change has a deep impact on the implementation of our techniques as the complexity of the derivatives of $\eta$ and $\theta$ increases drastically with the order (see Lemmas~\ref{Lemma:DerivativesKappa} and \ref{Lemma:DerivativesTheta}). As a result, it is unclear at the moment if our techniques could be effectively adapted to this case.
\end{remark}

\section{Additive Gaussian Mechanism}
\label{Section:BarronConstantPostProcessing}

In this section we consider the situation in which $Y$ is binary and $X\in\mathbb{R}$ is post-processed by the additive Gaussian mechanism introduced in Section~\ref{Section:PrivacyPreservingMechanisms}. Specifically, we assume that $X$ is post-processed to produce a new random variable
\begin{equation}
    X^{\sigma} \coloneqq \ X + \sigma Z,
\end{equation}
where $\sigma>0$ and $Z$ is a standard Gaussian random variable independent of $X$ and $Y$. Also, we assume that the random variable $X^{\sigma}$ is further processed to remove extreme values, giving rise to a random variable $\widetilde{X}^{\sigma}$. Specifically, we consider extreme values truncation and extreme values randomization as introduced in Sections~\ref{Subsection:ExtremeValuesTruncation} and \ref{Subsection:ExtremeValuesRandomization}, respectively. Given the different nature of these two processing techniques, below we provide estimates for the Barron constant of (i) the conditional expectation under truncation and (ii) the log-likelihood ratio under randomization.

\subsection{Extreme Values Truncation}

Consider extreme values truncation with $B = [-r,r]$ for some $r>0$. As before, we let $p \coloneqq \mathbb{P}(Y=1)$ and $\tilde{f}_{\pm}^{\sigma}$ be the conditional density of $\widetilde{X}^{\sigma}$ given $Y = \pm1$. The conditional expectation of $Y$ given $\widetilde{X}^{\sigma}$ is equal to
\begin{align}
    \label{eq:EVTtildeeta} \tilde{\eta}^{\sigma}(x) &= \frac{p\tilde{f}_{+}^{\sigma}(x) - \bar{p}\tilde{f}_{-}^{\sigma}(x)}{p\tilde{f}_{+}^{\sigma}(x) + \bar{p}\tilde{f}_{-}^{\sigma}(x)}\mathbbm{1}_{\lvert x \rvert \leq r}\\
    &= \tanh\left(\frac{1}{2} \log\left(\frac{p\tilde{f}_{+}^{\sigma}(x)}{\bar{p}\tilde{f}_{-}^{\sigma}(x)}\right)\right) \mathbbm{1}_{\lvert x \rvert \leq r}.
\end{align}
As established in \eqref{eq:DensityTruncation},
\begin{equation}
\label{eq:EVTtildef}
    \tilde{f}_{\pm}^{\sigma}(x) = \frac{(f_{\pm} \ast K_{\sigma})(x)}{\mathbb{P}\left(\lvert X^{\sigma} \rvert \leq r \vert Y=\pm1 \right)} \mathbbm{1}_{\lvert x \rvert \leq r},
\end{equation}
where $f_{\pm}$ is the conditional density of $X$ given $Y = \pm 1$ and
\begin{equation}
    K_{\sigma}(x) = \frac{1}{\sqrt{2\pi\sigma^2}} e^{-x^2/2\sigma^2}, \quad \quad x\in\mathbb{R}.
\end{equation}
From \eqref{eq:EVTtildef}, we conclude that $x\mapsto\frac{p\tilde{f}_{+}^{\sigma}(x)}{\bar{p}\tilde{f}_{-}^{\sigma}(x)}$ is a non-negative smooth function over $[-r,r]$ and, as a result, $\tilde{\eta}^{\sigma}$ is a smooth function over the same domain as well. As pointed out by Barron \cite[Sec.~IX]{barron1993universal}, this implies that $\tilde{\eta}^{\sigma}$ belongs to $\Gamma_{B}$ and, in particular, $C_{\tilde{\eta}^{\sigma}}$ is finite. Our goal is to find a tractable, yet useful, upper bound for $C_{\tilde{\eta}^{\sigma}}$.

As discussed in \eqref{eq:BarronExtension1}, a first step in order to find an upper bound for the Barron constant of $\tilde{\eta}^{\sigma}$ is to find a function, say $\eta^{\sigma}$, such that $\eta^{\sigma}$ is defined over $\mathbb{R}$ and $\tilde{\eta}^{\sigma} = \eta^{\sigma}\vert_{B}$. In this situation, we have that
\begin{equation}
\label{eq:BarronConstantEtaTilde}
    C_{\tilde{\eta}^{\sigma}} \leq \frac{r}{\sqrt{2\pi}} \int_{\mathbb{R}} \lvert\omega\rvert \lvert\widehat{\eta^{\sigma}}(\omega)\rvert \mathrm{d}\omega \eqqcolon r C_{\eta^{\sigma}}.
\end{equation}
Motivated by \eqref{eq:EVTtildeeta} and \eqref{eq:EVTtildef}, we define $\eta^{\sigma}:\mathbb{R}\to\mathbb{R}$ by
\begin{equation}
\label{eq:DefEtaSigma}
    \eta^{\sigma}(x) = \frac{\lambda_{+} f_{+}^{\sigma}(x) - \lambda_{-} f_{-}^{\sigma}(x)}{\lambda_{+} f_{+}^{\sigma}(x) + \lambda_{-} f_{-}^{\sigma}(x)},
\end{equation}
where $f_{\pm}^{\sigma} = f_{\pm} \ast K_{\sigma}$ and
\begin{equation}
    \lambda_{\pm} = \frac{\frac{1}{2} \pm \left(p-\frac{1}{2}\right)}{\mathbb{P}(\lvert X^{\sigma} \rvert \leq r \vert Y = \pm1)}.
\end{equation}
Note that, by large deviations arguments, $\lambda_{\pm}$ can be estimated with relatively high precision as it only depends on the probabilities of the events $\{Y=\pm1\}$ and $\{\lvert X^{\sigma} \rvert \leq r,Y=\pm1\}$. Furthermore, it can be shown that $\displaystyle \lim_{\sigma\to\infty} \tfrac{\lambda_{+}}{\lambda_{-}} = \tfrac{p}{\bar{p}}$, making the estimation of $\lambda_{\pm}$ unnecessary for large $\sigma$.

To gain some intuition about the behavior of the Barron constant as a function of $\sigma$, in the next proposition we compute $C_{\eta^{\sigma}}$ in a simple case.

\begin{proposition}
\label{Prop:BarronConstantGMM}
If $Y \sim \mathrm{Unif}(\{\pm1\})$ and $X = Y$, then,
\begin{equation}
    C_{\eta^{\sigma}} = \frac{1}{\sigma^{2}}, \quad \quad \sigma>0.
\end{equation}
\end{proposition}

\begin{proof}
By symmetry, we have that $\lambda_{+} = \lambda_{-}$. Thus, by \eqref{eq:DefEtaSigma},
\begin{equation}
    \eta^{\sigma}(x) = \frac{f_{+}^{\sigma}(x) - f_{-}^{\sigma}(x)}{f_{+}^{\sigma}(x) + f_{-}^{\sigma}(x)}.
\end{equation}
A direct computation shows that
\begin{equation}
    f_{\pm}^{\sigma}(x) = \frac{1}{\sqrt{2\pi\sigma^{2}}} e^{-(x\mp 1)^{2}/2\sigma^{2}}.
\end{equation}
Therefore, for all $x\in\mathbb{R}$,
\begin{equation}
\label{eq:ExperimentEta}
    \eta^{\sigma}(x) = \tanh\left(\frac{x}{\sigma^{2}}\right).
\end{equation}
Observe that $\displaystyle (\eta^{\sigma})'(x) = \frac{1}{\sigma^{2}} \mathrm{sech}^{2}\left(\frac{x}{\sigma^{2}}\right)$. Using contour integration, it can be verified that
\begin{equation}
\label{eq:BarronConstantGMMFT}
    \widehat{\textnormal{sech}^2}(\omega) = \sqrt{\frac{\pi}{2}} \omega\,\textnormal{csch}\left(\frac{\pi}{2}\omega\right).
\end{equation}
In particular, $(\eta^{\sigma})',\widehat{(\eta^{\sigma})'}\in L^{1}(\mathbb{R})$ and, by \eqref{eq:DefChDifferentiable},
\begin{equation}
\label{eq:RepresentationExampleBarronConstant}
    C_{\eta^{\sigma}} = \frac{1}{\sqrt{2\pi}} \int_\mathbb{R} \lvert\widehat{(\eta^{\sigma})'}(\omega)\rvert \mathrm{d}\omega.
\end{equation}
By \eqref{eq:BarronConstantGMMFT}, we have that $\widehat{(\eta^{\sigma})'}(\omega)$ is non-negative for all $\omega\in\mathbb{R}$. Therefore, by the Fourier inversion theorem, \eqref{eq:RepresentationExampleBarronConstant} implies that $\displaystyle C_{\eta^{\sigma}} = \frac{1}{\sigma^{2}} \textnormal{sech}^{2}(0)$.
\end{proof}

The next theorem provides an upper bound for the Barron constant of $\eta^{\sigma}$, as defined in \eqref{eq:DefEtaSigma}, under minimal assumptions on the distribution of $X$.

\begin{theorem}
\label{Theorem:MainTheoremPostProcessing}
If $\mathrm{Supp}(f_{\pm}) \subset [-1,1]$, then, for every $\sigma>0$,
\begin{equation}
    C_{\eta^{\sigma}} \leq \frac{2\sqrt{2}}{e\sqrt{\pi}} + \frac{16\sqrt{2}M_{0}^{\sigma}}{\sqrt{\pi}\sigma^{4}} \left(1 + \frac{1}{2} \log\left(\frac{M^{\sigma}}{\sigma^{8}}\right)\right),
\end{equation}
where
\begin{align}
    \label{eq:DefMomentsMp} M_{\alpha}^{\sigma} &\coloneqq \int_{\mathbb{R}} \lvert x \rvert^{\alpha} \frac{\lambda_{+} f_{+}^{\sigma}(x)\lambda_{-}f_{-}^{\sigma}(x)}{(\lambda_{+} f_{+}^{\sigma}(x)+\lambda_{-}f_{-}^{\sigma}(x))^{2}} \mathrm{d} x,\\
    M^{\sigma} &\coloneqq M_{0}^{\sigma} (64M_{2}^{\sigma} + 176M_{1}^{\sigma} + (136 + 48\sigma^{2})M_{0}^{\sigma}).
\end{align}
Furthermore, if $\sqrt[4]{8eM_{0}^{\sigma}} \leq \sigma$, then
\begin{equation}
\label{eq:MainTheoremPostProcessingLargeSigma}
    C_{\eta^{\sigma}} \leq \frac{16\sqrt{2}M_{0}^{\sigma}}{\sqrt{\pi}\sigma^{4}} \left(1 + \frac{1}{2} \log\left(\frac{M_{2}^{\sigma}}{M_{0}^{\sigma}}+3\frac{M_{1}^{\sigma}}{M_{0}^{\sigma}} + 3+\sigma^{2}\right)\right).
\end{equation}
\end{theorem}

The proof of the previous theorem, which can be found in Appendix~\ref{Appendix:MainTheoremPostProcessing}, relies on Theorem~\ref{Theorem:CEtaL1NormsBound} and careful estimates of the $L^{1}$-norms of the derivatives of $\eta^{\sigma}$. Specifically, we exploit the cancellations that occur between the terms in the numerators of \eqref{eq:Kappa0} -- \eqref{eq:Kappa2}.

Note that the bounds in the previous theorem only depend on the moment-like quantities $M_{\alpha}^{\sigma}$, as defined in \eqref{eq:DefMomentsMp}. As we show below, in some canonical situations $M_{\alpha}^{\sigma} = O(\sigma^{2(1+\alpha)})$ as $\sigma\to\infty$. Therefore, in the large noise regime ($\sigma\gg1$),
\begin{equation}
\label{eq:CEtaLargeNoiseRegime}
    C_{\eta^{\sigma}} \leq O\left(\frac{\log(\sigma)}{\sigma^{2}}\right).
\end{equation}
In view of Proposition~\ref{Prop:BarronConstantGMM}, we conclude that the previous bound is order optimal up to logarithmic factors. Below we also show that in some situations $M_{\alpha}^{\sigma} = O(1)$ as $\sigma\to0^{+}$. Therefore, in the small noise regime ($\sigma\ll1$),
\begin{equation}
\label{eq:CEtaSmallNoiseRegime}
    C_{\eta^{\sigma}} \leq O\left(\frac{\log(1/\sigma)}{\sigma^{4}}\right).
\end{equation}
Although the order optimality of this bound is unclear, it is by no means trivial. Observe that, as $\sigma\to0^{+}$, $\eta^{\sigma}$ converges pointwise to $\eta$ which in principle might have an unbounded Barron constant. Thus, \eqref{eq:CEtaSmallNoiseRegime} shows that even if $C_{\eta^{\sigma}}$ diverges to infinity as $\sigma\to0^{+}$, it does it polynomially in $1/\sigma$.

We end this section providing an upper bound for the moment-like quantities $M_{\alpha}^{\sigma}$ under different structural properties of the support of $f_{\pm}$. In the next proposition, we do so in the case where the supports of $f_{+}$ and $f_{-}$ are well-separated by some margin.

\begin{proposition}
\label{Proposition:MpSeparated}
Let $M_{\alpha}^{\sigma}$ be the quantities defined in \eqref{eq:DefMomentsMp}. If there exist $\gamma\in(0,1)$ such that $\mathrm{Supp}(f_{+}) \subset [\gamma,1]$ and $\mathrm{Supp}(f_{-}) \subset [-1,-\gamma]$, then, for every $\sigma>0$,
\begin{align}
    M_{0}^{\sigma} &\leq 2 + \frac{\sigma^{2}}{2\gamma} \frac{\lambda_{+}^{2}+\lambda_{-}^{2}}{\lambda_{+}\lambda_{-}} e^{-2\gamma/\sigma^{2}},\\
    M_{1}^{\sigma} &\leq 2 + \left(\frac{\sigma^{4}}{4\gamma^{2}} + \frac{\sigma^{2}}{2\gamma}\right) \frac{\lambda_{+}^{2}+\lambda_{-}^{2}}{\lambda_{+}\lambda_{-}} e^{-2\gamma/\sigma^{2}},\\
    M_{2}^{\sigma} &\leq 2 + \left(\frac{\sigma^{6}}{4\gamma^{3}} + \frac{\sigma^{4}}{2\gamma^{2}} + \frac{\sigma^{2}}{2\gamma}\right) \frac{\lambda_{+}^{2}+\lambda_{-}^{2}}{\lambda_{+}\lambda_{-}} e^{-2\gamma/\sigma^{2}}.
\end{align}
In particular, $M_{\alpha}^{\sigma} = O(\sigma^{2(1+\alpha)})$ as $\sigma\to\infty$ and $M_{\alpha}^{\sigma} = O(1)$ as $\sigma\to0^{+}$.
\end{proposition}

\begin{proof}
See Appendix~\ref{Appendix:ProofMPSeparated}.
\end{proof}

In the next proposition we provide upper bounds for $M_{\alpha}^{\sigma}$ in the case where the supports of $f_{\pm}$ overlap but extreme values determine the value of $Y$, i.e., there exists $\gamma_{0}$ such that if $X > \gamma_{0}$ then $Y = 1$, and if $X < -\gamma_{0}$ then $Y = -1$.

\begin{proposition}
\label{Proposition:MpOverlapping}
Let $M_{\alpha}^{\sigma}$ be the quantities defined in \eqref{eq:DefMomentsMp}. If there exist $\gamma_{0}\in(0,1)$ such that
\begin{align}
    [\gamma_{0},1] &\subset \mathrm{Supp}(f_{+}) \subset (-\gamma_{0},1],\\
    [-1,-\gamma_{0}] &\subset \mathrm{Supp}(f_{-}) \subset [-1,\gamma_{0}),
\end{align}
then, for every $\gamma\in(\gamma_{0},1)$ and $\sigma>0$,
\begin{align}
    M_{0}^{\sigma} &\leq 2 + \frac{\sigma^{2}}{\gamma-\gamma_{0}} \Lambda,\\
    M_{1}^{\sigma} &\leq 2 + \left(\frac{\sigma^{4}}{(\gamma-\gamma_{0})^{2}} + \frac{\sigma^{2}}{\gamma-\gamma_{0}}\right) \Lambda,\\
    M_{2}^{\sigma} &\leq 2 + \left(\frac{2\sigma^{6}}{(\gamma-\gamma_{0})^{3}} + \frac{2\sigma^{4}}{(\gamma-\gamma_{0})^{2}} + \frac{\sigma^{2}}{\gamma-\gamma_{0}}\right) \Lambda,
\end{align}
where $\displaystyle \Lambda = \frac{\delta_{+}\lambda_{+}^{2}+\delta_{-}\lambda_{-}^{2}}{\delta_{+}\lambda_{+}\delta_{-}\lambda_{-}}$,
\begin{equation}
    \delta_{+} = \int_{\gamma}^{1} f_{+}(s) \mathrm{d} s \quad \textnormal{and} \quad \delta_{-} = \int_{-1}^{-\gamma} f_{+}(s) \mathrm{d} s.
\end{equation}
In particular, $M_{\alpha}^{\sigma} = O(\sigma^{2(1+\alpha)})$ as $\sigma\to\infty$ and $M_{\alpha}^{\sigma} = O(1)$ as $\sigma\to0^{+}$.
\end{proposition}

\begin{proof}
See Appendix~\ref{Appendix:ProofMPOverlapping}.
\end{proof}

\subsection{Extreme Values Randomization}

Consider extreme values randomization with $B = [-r,r]$ for some $r>0$. As before, we let $p \coloneqq \mathbb{P}(Y=1)$ and $\tilde{f}_{\pm}^{\sigma}$ be the conditional density of $\widetilde{X}^{\sigma}$ given $Y = \pm1$. Recall the definition of the log-likelihood function
\begin{align}
    \label{eq:EVRTildeTheta} \tilde{\theta}^{\sigma}(x) = \frac{1}{2} \log\left(\frac{p\tilde{f}_{+}^{\sigma}(x)}{\bar{p}\tilde{f}_{-}^{\sigma}(x)}\right) \mathbbm{1}_{\lvert x \rvert \leq r}.
\end{align}
As established in \eqref{eq:DensityRandomization},
\begin{equation}
\label{eq:EVRTildeF}
    \tilde{f}_{\pm}^{\sigma}(x) = \left[(f_{\pm} \ast K_{\sigma})(x) + \frac{\mathbb{P}\left(\lvert X^{\sigma}\rvert > r \vert Y=\pm1 \right)}{2r}\right] \mathbbm{1}_{\lvert x \rvert \leq r},
\end{equation}
where $f_{\pm}$ is the conditional density of $X$ given $Y = \pm 1$ and
\begin{equation}
    K_{\sigma}(x) = \frac{1}{\sqrt{2\pi\sigma^2}} e^{-x^2/2\sigma^2}, \quad \quad x\in\mathbb{R}.
\end{equation}
From \eqref{eq:EVRTildeF}, we conclude that $x\mapsto\frac{p\tilde{f}_{+}^{\sigma}(x)}{\bar{p}\tilde{f}_{-}^{\sigma}(x)}$ is a positive smooth function over $[-r,r]$ and, as a result, $\tilde{\theta}^{\sigma}$ is a smooth function over the same domain as well. As pointed out by Barron \cite[Sec.~IX]{barron1993universal}, this implies that $\tilde{\theta}^{\sigma}$ belongs to $\Gamma_{B}$ and, in particular, $C_{\tilde{\theta}^{\sigma}}$ is finite. Our goal is to find a tractable, yet useful, upper bound for $C_{\tilde{\theta}^{\sigma}}$.

As discussed in \eqref{eq:BarronExtension1}, a first step in order to find an upper bound for the Barron constant of $\tilde{\theta}^{\sigma}$ is to find a function, say $\theta^{\sigma}$, such that $\theta^{\sigma}$ is defined over $\mathbb{R}$ and $\tilde{\theta}^{\sigma} = \theta^{\sigma}\vert_{B}$. In this situation, we have that
\begin{equation}
\label{eq:BarronConstantThetaTilde}
    C_{\tilde{\theta}^{\sigma}} \leq \frac{r}{\sqrt{2\pi}} \int_{\mathbb{R}} \lvert\omega\rvert \lvert\widehat{\theta^{\sigma}}(\omega)\rvert \mathrm{d}\omega \eqqcolon r C_{\theta^{\sigma}}.
\end{equation}
Motivated by \eqref{eq:EVRTildeTheta} and \eqref{eq:EVRTildeF}, we define $\theta^{\sigma}:\mathbb{R}\to\mathbb{R}$ by
\begin{equation}
\label{eq:DefThetaSigma}
    \theta^{\sigma}(x) = \frac{1}{2} \log\left(\frac{pf_{+}^{\sigma}(x) + \lambda_{+}}{\bar{p}f_{-}^{\sigma}(x) + \lambda_{-}}\right),
\end{equation}
where $f_{\pm}^{\sigma} = f_{\pm} \ast K_{\sigma}$ and
\begin{equation}
\label{eq:DefLambdaTheta}
    \lambda_{\pm} = \frac{\mathbb{P}\left(\lvert X^{\sigma}\rvert > r \vert Y=\pm1 \right)}{2r}.
\end{equation}
Note that, by large deviations arguments, $\lambda_{\pm}$ can be estimated with relatively high precision as it only depends on the proba\-bilities of the events $\{Y=\pm1\}$ and $\{\lvert X^{\sigma} \rvert > r,Y=\pm1\}$. Furthermore, it can be shown that $\displaystyle \lim_{\sigma\to\infty} \lambda_{\pm} = \tfrac{1}{2r}$, making the estimation of $\lambda_{\pm}$ unnecessary for large $\sigma$.

The next theorem provides an upper bound for the Barron constant of $\theta^{\sigma}$, as defined in \eqref{eq:DefThetaSigma}, under minimal assumptions on the distribution of $X$.

\begin{theorem}
\label{Theorem:MainTheoremPostProcessingTheta}
If $f_{\pm}$ is a probability density function, then, for every $\sigma>0$,
\begin{equation}
\label{eq:MainTheoremPostProcessingTheta}
    C_{\theta^{\sigma}} \leq \frac{2\sqrt{2}}{e\sqrt{\pi}} + \frac{2\sqrt{2}N_{2}^{\sigma}}{\sqrt{\pi}} \left(1 + \frac{1}{2}\log\left(N_{1}^{\sigma} N_{3}^{\sigma} \right)\right),
\end{equation}
where
\begin{align}
    \label{eq:DefN1Sigma} N_{1}^{\sigma} &\coloneqq \frac{\Lambda_{1}}{\sqrt{2\pi}\sigma},\\
    \label{eq:DefN2Sigma} N_{2}^{\sigma} &\coloneqq \frac{\Lambda_{1}}{\sigma^{2}} + \frac{\Lambda_{2}}{8\sqrt{\pi}\sigma^{3}},\\
    \label{eq:DefN3Sigma} N_{3}^{\sigma} &\coloneqq \frac{5\Lambda_{1}}{\sqrt{2\pi}\sigma^{3}} + \frac{3\sqrt{3}\Lambda_{2}}{8\sqrt{2\pi}\sigma^{4}} + \frac{\sqrt{2}\Lambda_{3}}{9\sqrt{\pi^{3}}\sigma^{5}},
\end{align}
with $\displaystyle \Lambda_{\alpha} = \frac{p^{\alpha}}{\lambda_{+}^{\alpha}} + \frac{\bar{p}^{\alpha}}{\lambda_{-}^{\alpha}}$. Moreover, if $N_{2}^{\sigma} \leq 1/e$, then
\begin{equation}
\label{eq:MainTheoremPostProcessingThetaLargeSigma}
    C_{\theta^{\sigma}} \leq \frac{2\sqrt{2}N_{2}^{\sigma}}{\sqrt{\pi}} \left(1 + \frac{1}{2} \log\left(\frac{N_{1}^{\sigma}N_{3}^{\sigma}}{(N_{2}^{\sigma})^{2}}\right)\right).
\end{equation}
\end{theorem}

The proof of the previous theorem, which can be found in Appendix~\ref{Appendix:MainTheoremPostProcessingTheta}, relies on Theorem~\ref{Theorem:CEtaL1NormsBound} and careful estimates of the $L^{1}$-norms of the derivatives of $\theta^{\sigma}$.

Note that the previous theorem does not assume anything about $f_{\pm}$ apart from its existence. Furthermore, \eqref{eq:MainTheoremPostProcessingThetaLargeSigma} implies that, in the large noise regime ($\sigma \gg 1$),
\begin{equation}
\label{eq:CThetaLargeNoiseRegime}
    C_{\theta^{\sigma}} = O\left(\frac{1}{\sigma^{2}}\right).
\end{equation}
By \cite[Sec.~IX-14]{barron1993universal}, it can be verified that in the context of Proposition~\ref{Prop:BarronConstantGMM} with $r=\infty$, we have $C_{\theta^{\sigma}} = 1/\sigma^{2}$. Therefore, the previous bound is in fact order optimal. Similarly, \eqref{eq:MainTheoremPostProcessingThetaLargeSigma} implies that, in the small noise regime ($\sigma \ll 1$),
\begin{equation}
\label{eq:CThetaSmallNoiseRegime}
    C_{\theta^{\sigma}} = O\left(\frac{\log(1/\sigma)}{\sigma^{3}}\right).
\end{equation}
As with \eqref{eq:CEtaSmallNoiseRegime}, it is unclear if \eqref{eq:CThetaSmallNoiseRegime} is order optimal.

\section{Numerical Considerations}
\label{Section:NumericalConsiderations}

In this section we explore some numerical aspects of our lower bounds for the MMSE. Specifically, we evaluate the upper bounds for the Barron constant produced by Theorems~\ref{Theorem:MainTheoremPostProcessing} and \ref{Theorem:MainTheoremPostProcessingTheta} in a particular setting. Also, we study the effect of the values of $k$ and $n$ on the proposed estimator, and propose an optimization method to approximate the value of $\mathrm{mmse}_{k,n}^{\ast}(Y \vert X)$ efficiently. We finish this section with a numerical illustration of our lower bounds for the MMSE.

\subsection{Upper Bounds for the Barron Constant}
\label{Section:NumericalConsiderationsBarronConstant}

In this section we evaluate the upper bounds for the Barron constant derived in Theorems~\ref{Theorem:MainTheoremPostProcessing} and \ref{Theorem:MainTheoremPostProcessingTheta}. To this end, we consider the setting where $Y \sim \textnormal{Unif}(\{\pm1\})$, $X = Y$ and $r=2$.

Observe that under extreme value truncation, \eqref{eq:BarronConstantEtaTilde} and Proposition~\ref{Prop:BarronConstantGMM} imply that
\begin{equation}
\label{eq:EvaluationBarronConstantProp2}
    C_{\tilde{\eta}^{\sigma}} \leq \frac{2}{\sigma^{2}}.
\end{equation}
Since this bound depends on the exact computation performed in Proposition~\ref{Prop:BarronConstantGMM}, we use \eqref{eq:EvaluationBarronConstantProp2} as a benchmark for the upper bounds obtained using Theorems~\ref{Theorem:MainTheoremPostProcessing} and \ref{Theorem:MainTheoremPostProcessingTheta}.

Note that in the current setting the assumption of Proposition~\ref{Proposition:MpSeparated} is satisfied with $\gamma=1$. Also, by symmetry,
\begin{equation}
    \mathbb{P}(\lvert X^{\sigma} \rvert \leq r \vert Y = +1) = \mathbb{P}(\lvert X^{\sigma} \rvert \leq r \vert Y = -1),
\end{equation}
which implies that $\lambda_{+}=\lambda_{-}$. Thus, Proposition~\ref{Proposition:MpSeparated} leads to
\begin{align}
    M_{0}^{\sigma} &\leq 2 + \sigma^{2} e^{-2/\sigma^{2}},\\
    M_{1}^{\sigma} &\leq 2+\left(\frac{\sigma^{4}}{2}+\sigma^{2}\right) e^{-2/\sigma^{2}},\\
    M_{2}^{\sigma} &\leq 2\left(\frac{\sigma^{6}}{2} + \sigma^{4} + \sigma^{2}\right) e^{-2/\sigma^{2}}.
\end{align}
Finally, \eqref{eq:BarronConstantEtaTilde} and Theorem~\ref{Theorem:MainTheoremPostProcessing} lead to the upper bound
\begin{equation}
\label{eq:EvaluationBarronConstantThm5}
    C_{\tilde{\eta}^{\sigma}} \leq \frac{32\sqrt{2}M_{0}^{\sigma}}{\sqrt{\pi}\sigma^{4}} \left(1 + \frac{1}{2} \log\left(\frac{M_{2}^{\sigma}}{M_{0}^{\sigma}}+3\frac{M_{1}^{\sigma}}{M_{0}^{\sigma}} + 3+\sigma^{2}\right)\right),
\end{equation}
which holds true whenever $8e(2+\sigma^{2}e^{-2/\sigma^{2}})\leq\sigma^{4}$.  It can be verified that the previous inequality holds for $\sigma\geq4.7$.

In the current setting under extreme value randomization, it can be shown that
\begin{equation}
\label{eq:EvaluationTheorem6}
    \mathbb{P}\left(\lvert X^{\sigma} \rvert > 2 \vert Y = \pm1 \right) = Q(1/\sigma) + Q(3/\sigma),
\end{equation}
where $\displaystyle Q(x) = \frac{1}{\sqrt{2\pi}} \int_{x}^{\infty} e^{-x^{2}/2} \mathrm{d}x$. In particular, in the notation of Theorem~\ref{Theorem:MainTheoremPostProcessingTheta}, we have that
\begin{equation}
    \lambda_{\pm} = \frac{Q(1/\sigma) + Q(3/\sigma)}{4},
\end{equation}
and
\begin{equation}
    \Lambda_{\alpha} = \frac{2^{\alpha+1}}{(Q(1/\sigma) + Q(3/\sigma))^{\alpha}}.
\end{equation}
Using the previous expressions, we can provide upper bounds for $N_{1}^{\sigma}$, $N_{2}^{\sigma}$ and $N_{3}^{\sigma}$ as defined in Theorem~\ref{Theorem:MainTheoremPostProcessingTheta}. The latter theorem and \eqref{eq:BarronConstantThetaTilde} lead to
\begin{equation}
\label{eq:EvaluationBarronConstantThm6}
    C_{\tilde{\theta}^{\sigma}} \leq \frac{4\sqrt{2}N_{2}^{\sigma}}{\sqrt{\pi}} \left(1 + \frac{1}{2} \log\left(\frac{N_{1}^{\sigma}N_{3}^{\sigma}}{(N_{2}^{\sigma})^{2}}\right)\right),
\end{equation}
which holds true whenever $\frac{\Lambda_{1}}{\sigma^{2}}+\frac{\Lambda_{2}}{8\sqrt{\pi}\sigma^{3}} \leq \frac{1}{e}$. It can be verified that the previous inequality holds for $\sigma\geq4.25$.

The bounds \eqref{eq:EvaluationBarronConstantProp2}, \eqref{eq:EvaluationBarronConstantThm5} and \eqref{eq:EvaluationBarronConstantThm6} are illustrated in Figure~\ref{Fig:BarronConstantBounds}. We would like to remark that, in order to evaluate the bound produced by Theorem~\ref{Theorem:MainTheoremPostProcessingTheta}, we used the fact that $\gamma=1$ to obtain the exact probability in \eqref{eq:EvaluationTheorem6}. While this assumption is rather strong as it amounts to know that $X=Y$, a similar assumption was made to evaluate the bound produced by Theorem~\ref{Theorem:MainTheoremPostProcessing}. Hence, the comparison of these bounds is fair and suggests that Theorem~\ref{Theorem:MainTheoremPostProcessingTheta} produces better bounds in practice than Theorem~\ref{Theorem:MainTheoremPostProcessing}. Hence, we focus on the numerical evaluation of $\mathrm{mmse}_{k,n}^{\ast}(Y \vert X)$ for the remainder of this section.

\begin{figure}
    \centering
    \includegraphics[width=0.3\textwidth]{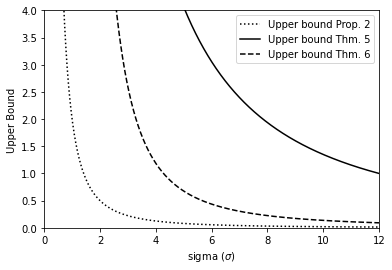}
    \caption{Bounds for the Barron constant produced by Proposition~\ref{Prop:BarronConstantGMM}, Theorem~\ref{Theorem:MainTheoremPostProcessing} and Theorem~\ref{Theorem:MainTheoremPostProcessingTheta}. These bounds hold for $\sigma\geq0$, $\sigma\geq4.7$, and $\sigma\geq4.25$, respectively. The bound produced by Proposition~\ref{Prop:BarronConstantGMM} depends on an exact computation and serves as a benchmark. Overall, Theorem~\ref{Theorem:MainTheoremPostProcessingTheta} seems to produce better bounds in practice than Theorem~\ref{Theorem:MainTheoremPostProcessing}.}
    \label{Fig:BarronConstantBounds}
\end{figure}

\subsection{Computation of $\mathrm{mmse}_{k,n}^{\ast}(Y \vert X)$}

A key difficulty to instantiate the proposed lower bounds for the MMSE is to determine the appropriate values of $k$ and $n$. In this section we study this problem and propose an optimization method to approximate the value of $\mathrm{mmse}_{k,n}^{\ast}(Y \vert X)$ efficiently.

For each $k\in\mathbb{N}$, let $\mathcal{S}_{k}$ be the set of functions $g:\mathbb{R}\to\mathbb{R}$ of the form
\begin{equation}
\label{eq:GeneralFormSk}
    g(x) = \sum_{l=0}^{k} s_{l} \mathbbm{1}_{[t_{l},t_{l+1})}(x),
\end{equation}
where $s_{l}\in\{-1,0,+1\}$ and
\begin{equation}
    -\infty \equiv t_{0} < t_{1} < \cdots < t_{k} < t_{k+1} \equiv \infty.
\end{equation}
We refer to $t_{1},\ldots,t_{k}$ as the threshold points of $g$. Observe that $\mathcal{S}_{k}$ captures the set of \emph{class probability estimators} \cite{reid2010composite} that are either completely confident (i.e., $g(x)=\pm1$) or completely uncertain (i.e., $g(x)=0$) about their predictions.

The next lemma establishes a key structural property of the hypothesis class $\tanh\circ\mathcal{H}_{k}^{\phi}$: any function $g\in\mathcal{S}_{k}$ can be approximated, in the uniform norm outside a neighborhood of the threshold points of $g$, by functions on $\tanh\circ\mathcal{H}_{k}^{\phi}$.

\begin{lemma}
\label{Lemma:StructuralPropertyHk}
Let $k\in\mathbb{N}$ and $\epsilon,\delta>0$. If $g\in\mathcal{S}_{k}$, then there exists $h\in\mathcal{H}_{k}^{\phi}$ such that, for all $x\in\mathbb{R}\setminus\bigcup_{l=1}^{k} (t_{l}-\delta,t_{l}+\delta)$,
\begin{equation}
\label{eq:ThmStructurlProperty}
    \lvert g(x) - \tanh(h(x)) \lvert \leq \epsilon.
\end{equation}
\end{lemma}

The proof of the previous lemma, which can be found in Appendix~\ref{Appendix:ProofLemmaStructuralPropertyHk}, relies on standard (uniform) approximation arguments. The next theorem establishes that the minimum empirical square-loss over $\mathcal{S}_{k}$ serves as an upper bound for $\mathrm{mmse}_{k,n}^{\ast}(Y \vert X)$.

\begin{theorem}
\label{Theorem:StructuralPropertyHk}
If $k\in\mathbb{N}$, then
\begin{equation}
    \mathrm{mmse}_{k,n}^{\ast}(Y \vert X) \leq \inf_{g\in\mathcal{S}_{k}} \frac{1}{n} \sum_{i=1}^{n} (Y_{i}-g(X_{i}))^{2}.
\end{equation}
\end{theorem}

\begin{proof}
Let $g\in\mathcal{S}_{k}$ be given as in \eqref{eq:GeneralFormSk}. For each $l\in[k]$, let
\begin{equation}
    \tilde{t}_{l} \coloneqq t_{l} - \frac{1}{2} \left(1 \wedge \min\left\{t_{l} - X_{i} : X_{i} < t_{l}\right\}\right),
\end{equation}
where we take the minimum of the empty set as $+\infty$. We define the function $\tilde{g}:\mathbb{R}\to\mathbb{R}$ as
\begin{equation}
    \tilde{g}(x) \coloneqq \sum_{l=0}^{k} s_{l} \mathbbm{1}_{[\tilde{t}_{l},\tilde{t}_{l+1})}(x),
\end{equation}
where $\tilde{t}_{0} = - \infty$ and $\tilde{t}_{k+1} = \infty$. It can be verified that $\tilde{g}\in\mathcal{S}_{k}$, $\{t_{1},\ldots,t_{k}\} \cap \{X_{1},\ldots,X_{n}\} = \emptyset$ and
\begin{equation}
\label{eq:ProofMMSESkEqualityTilde}
    \frac{1}{n} \sum_{i=1}^{n} (Y_{i}-g(X_{i}))^{2} = \frac{1}{n} \sum_{i=1}^{n} (Y_{i}-\tilde{g}(X_{i}))^{2}.
\end{equation}
Let $\epsilon>0$. Take $\delta>0$ such that, for every $i\in[n]$,
\begin{equation}
    X_{i} \notin \bigcup_{l=1}^{k} (\tilde{t}_{l}-\delta,\tilde{t}_{l}+\delta).
\end{equation}
By Lemma~\ref{Lemma:StructuralPropertyHk}, there exists $h\in\mathcal{H}_{k}^{\phi}$ such that, for every $i\in[n]$,
\begin{equation}
    \lvert \tilde{g}(X_{i}) - \tanh(h(X_{i})) \lvert \leq \epsilon.
\end{equation}
The previous inequality and the triangle inequality lead to
\begin{align}
    \sum_{i=1}^{n} (Y_{i} - \tanh(h(X_{i})))^{2} &\leq n\epsilon^{2} + 2\epsilon \sum_{i=1}^{n} \lvert Y_{i} - \tilde{g}(X_{i})\rvert + \sum_{i=1}^{n} (Y_{i} - \tilde{g}(X_{i}))^{2}.
\end{align}
Therefore, \eqref{eq:ProofMMSESkEqualityTilde} and the fact that $\epsilon>0$ is arbitrary imply that
\begin{equation}
    \mathrm{mmse}_{k,n}^{\ast}(Y \vert X) \leq \frac{1}{n} \sum_{i=1}^{n} (Y_{i} - g(X_{i}))^{2}.
\end{equation}
Since $g\in\mathcal{S}_{k}$ is also arbitrary, the conclusion follows.
\end{proof}

The next corollary is a straightforward consequence of the previous theorem. Indeed, it follows by taking $g\in\mathcal{S}_{k}$ such that $g(X_{i}) = Y_{i}$ for all $i\leq k$ and $g(X_{i})=0$ for all $i>k$.

\begin{corollary}
\label{Corollary:UpperBoundMMSEkn}
If $k,n\in\mathbb{N}$ with $k\leq n$, then
\begin{equation}
\label{eq:PropSimpleUpperBoundMMSEkn}
    \mathrm{mmse}_{k,n}^{\ast}(Y \vert X) \leq 1 - \frac{k}{n}.
\end{equation}
\end{corollary}

From \eqref{eq:PropSimpleUpperBoundMMSEkn} we conclude that it is necessary to have $k \ll n$ in order to obtain meaningful bounds for $\mathrm{mmse}(Y \vert X)$. Note that the previous bound recovers the well-known fact that a two-layer neural network of size $k$ can memorize an entire sample of size $n$ whenever $k\geq n$, see, e.g., \cite{bubeck2020network}.

Motivated by Theorem~\ref{Theorem:StructuralPropertyHk}, we propose the following optimization process to approximate the value of $\mathrm{mmse}_{k,n}^{\ast}(Y \vert X)$: minimize the empirical square-loss around 0 using random initialization and gradient descent; minimize the empirical square-loss over $\mathcal{S}_{k}$ using dynamic programming, as described in\footnote{We implicitly assume that $X_{1},\ldots,X_{n}$ are pairwise different, which is the case in most practical cases, e.g., when the distribution of $X$ is absolutely continuous with respect to the Lebesgue measure.} Algorithm~\ref{Algorithm:Optimization}; and take the minimum of those two empirical losses. While this combined minimization process is not guaranteed to find the exact value of $\mathrm{mmse}_{k,n}^{\ast}(Y \vert X)$, it covers two important subsets of the hypothesis class $\tanh\circ\mathcal{H}_{k}^{\phi}$.

\begin{algorithm}[t]
\begingroup
\small
\caption{Empirical Square-Loss Minimization over $\mathcal{S}_{k}$}
\label{Algorithm:Optimization}

\begin{algorithmic}[1]
\State compute a permutation $\pi$ such that $X_{\pi(1)} < \cdots < X_{\pi(n)}$


\State set $L[l,s,i] = \infty$ for $0 \leq l \leq k$, $s\in\{-1,0,+1\}$, $0 \leq i \leq n$\Comment{minimal loss using $l$ thresholds up to $X_{\pi(i)}$ with $s_{l}=s$}

\State set $L[0,s,0]=0$ for $s\in\{-1,0,+1\}$ \Comment{setting $s_{0} = s$}

\For{$i = 1,\ldots,n$}
    \State $L[0,s,i] = L[0,s,i-1] + (Y_{\pi(i)} - s)^{2}$ for $s\in\{-1,0,+1\}$
    
    \For{$l = 1,\ldots,k$}
        
        \State $\displaystyle L[l,s,i] = (L[l,s,i-1] \wedge \min_{s'\neq s}L[l-1,s',i-1]) + (Y_{\pi(i)}-s)^{2}$ for $s\in\{-1,0,+1\}$
    \EndFor
\EndFor

\State \textbf{return} $\displaystyle \frac{1}{n} \min_{l,s} L[l,s,n]$ \Comment{minimal loss over $\mathcal{S}_{k}$}
\end{algorithmic}
\endgroup
\end{algorithm}
%

\subsection{Numerical Experiment}
\label{Subsection:NumericalConsiderationsExperiment}

We end this section applying the tools developed so far in a concrete numerical example. We consider the setting introduced in Section~\ref{Section:NumericalConsiderationsBarronConstant}, where $Y \sim \textnormal{Unif}(\{\pm1\})$, $X = Y$ and $r=2$. As pointed out in that section, in this setting Theorem~\ref{Theorem:MainTheoremPostProcessingTheta} produces better bounds for the Barron constant than Theorem~\ref{Theorem:MainTheoremPostProcessing}. Hence, for the sake of illustration, we focus on the extreme values randomization setting introduced in Section~\ref{Subsection:ExtremeValuesRandomization}.

Motivated by Corollary~\ref{Corollary:UpperBoundMMSEkn}, in our numerical experiments we set $k = n/100$ for $n=10,000$ and $n=100,000$. Recall that our optimization strategy to approximate the value $\mathrm{mmse}_{k,n}^{\ast}(Y \vert X)$ consists in (a) minimize the empirical square-loss around 0 using random initialization and gradient descent\footnote{We initialized the weights of the neural network at random with distribution $\mathcal{N}(0,0.01)$. When a random initialization with empirical square-loss less than 1 was found, 100 iteration of gradient descent with step size equal to 0.1 were performed. For each value of $\sigma$, this experiment was conducted 5 times and the best set of parameters was stored.}; (b) minimize the empirical square-loss over $\mathcal{S}_{k}$ using Algorithm~\ref{Algorithm:Optimization}; and (c) take the minimum of those two empirical losses. In all of our experiments, the minimal empirical square-loss over $\mathcal{S}_{k}$ was no larger than $0.9555$, while the minimal empirical square-loss around 0 was no smaller than $0.9997$. Thus, Algorithm~\ref{Algorithm:Optimization} seems to perform significantly better than standard machine learning techniques for the task of minimizing the empirical square-loss.

In Figure~\ref{Figure:NumericalExperiment} we plot our numerical results for $n=10,000$ and $n=100,000$, and a variety of values of $\sigma$. Note that the quality of the lower bound for the MMSE improves as $n$ and $k$ increase. However, as suggested by Corollary~\ref{Corollary:UpperBoundMMSEkn}, the ratio between $k$ and $n$ should remain bounded from below in order to get a meaningful bound.

We conjecture that the family $\mathcal{S}_{k}$ contains functions with relatively small empirical square-loss in the regime where $k \ll n$. This seems to be the case since $\mathcal{S}_{k}$ models the functions in $\tanh\circ\mathcal{H}_{k}^{\phi}$ that highly overfit to a portion of the data. Since Algorithm~\ref{Algorithm:Optimization} has complexity $O(kn)$, the minimal empirical square-loss over $\mathcal{S}_{k}$ provides a reasonable proxy for $\mathrm{mmse}_{k,n}^{\ast}(Y \vert X)$ that can be computed efficiently.

\begin{figure}[b]
	\centering
    \includegraphics[width=0.24\textwidth]{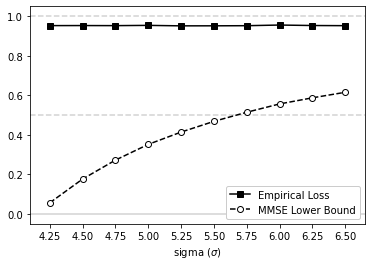} \qquad \qquad \qquad \includegraphics[width=0.24\textwidth]{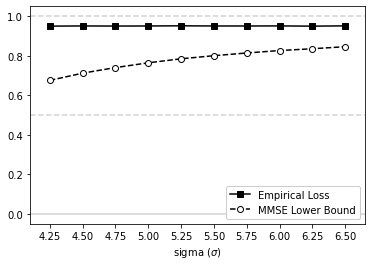}
	\caption{Minimal empirical square-loss over $\mathcal{S}_{k}$, which was significantly smaller than the corresponding loss around 0 in all of our experiments, and the lower bound for the MMSE obtained from Theorem~\ref{Theorem:2LNNtanh} and the bound for the Barron constant in Section~\ref{Section:NumericalConsiderationsBarronConstant}. We performed our experiments for $n=10,000$ (left) and $n=100,000$ (right) with $k = n/100$.}
	\label{Figure:NumericalExperiment}
\end{figure}

\section{Summary and Final Remarks}
\label{Section:Conclusion}

Motivated by estimation-theoretic privacy, in this paper we have established provable lower bounds for the MMSE in estimating a random variable $Y\in\mathbb{R}$ given another random variable $X\in\mathbb{R}^{d}$ (Theorems~\ref{Theorem:2LNNIdentity} and \ref{Theorem:2LNNtanh}). These bounds are based on a two-layer neural network estimator of the MMSE and the Barron constant of an appropriate function of the conditional expectation of $Y$ given $X$. More specifically, we have proposed the minimum empirical square-loss attained by a two-layer neural network of size $k$ as an estimator of the MMSE. We considered two variations of this estimator: the first one, denoted by $\mathrm{mmse}_{k,n}(Y \vert X)$, uses the identity function as the output activation function; while the second one, denoted by $\mathrm{mmse}^{\ast}_{k,n}(Y \vert X)$, uses hyperbolic tangent as the output activation function.

Finding meaningful estimates for the Barron constant is challenging since (i) the underlying conditional expectation is rarely available in practice and (ii) the Barron constant is defined in terms of the Fourier transform of this conditional expectation. To alleviate the second issue, we provided an upper bound for the Barron constant of a function $h:\mathbb{R}\to\mathbb{R}$ based on the $L^{1}$-norms of its derivatives (Theorem~\ref{Theorem:CEtaL1NormsBound}). We have further generalized this result to multivariate functions $h:\mathbb{R}^{d}\to\mathbb{R}$ (Theorem~\ref{Theorem:CEtaL1NormsBoundHighDimensions}), although the complexity of the result and the bound itself increase exponentially with $d$. In addition, we have shown that one can circumvent the first issue in applications where the additive Gaussian mechanism is used (Theorems~\ref{Theorem:MainTheoremPostProcessing} and \ref{Theorem:MainTheoremPostProcessingTheta}). In such applications, our estimates for the Barron constant are order optimal in the large noise regime.

In order to obtain numerical lower bounds for the MMSE in some concrete applications, we analyzed some algorithmic aspects related to the computation of the proposed estimator. First, we empirically found that the bounds for the Barron constant associated with the estimator $\mathrm{mmse}^{\ast}_{k,n}(Y \vert X)$ are tighter than those corresponding to $\mathrm{mmse}_{k,n}(Y \vert X)$. Building upon a structural property of the hypothesis class $\tanh\circ\mathcal{H}_{k}^{\phi}$ (Lemma~\ref{Lemma:StructuralPropertyHk} and Theorem~\ref{Theorem:StructuralPropertyHk}), we showed that the neural network size $k$ should be significantly smaller than the sample size $n$ in order to obtain meaningful lower bounds for the MMSE (Corollary~\ref{Corollary:UpperBoundMMSEkn}). Moreover, motivated by the same structural property, we proposed an optimization process to approximate the value of $\mathrm{mmse}^{\ast}_{k,n}(Y \vert X)$ that performs better than standard machine learning techniques and that can be computed efficiently using dynamic programming. Overall, we developed an effective machinery to obtain theoretical lower bounds for the MMSE.

While we have only considered shallow neural networks, there are fundamental obstructions in trying to generalize the present work to deep neural networks.
\begin{itemize}
    \item From a function approximation perspective, at the moment it seems that there is no analogue of Barron's theorem for deep neural networks\footnote{Lee \textit{et al.} \cite{lee2017ability} have an important effort in this direction, although their results depend on a specific decomposition of the target function.}. While there are many results explaining the approximation power of deep neural networks, see, e.g., \cite{lin2017does,liang2017deep,rolnick2018power}, they are mainly qualitative and, hence, unfitted to produce concrete bounds.
    
    \item From a computational perspective, the optimization landscape of deep neural networks is significantly more complex than its shallow counterpart, see, e.g., \cite{kawaguchi2016deep,yun2018small,li2018visualizing}. As a result, it is harder to guarantee that a deep neural network has been trained to optimality, which is essential for the estimator proposed in this paper.
    
    \item As mentioned at the end of Section~\ref{Section:Estimation2LNN}, if $k$ is sufficiently large then our lower bounds for the MMSE become trivial due to overfitting, i.e., $\mathrm{mmse}_{k,n}(Y \vert X)$ being equal to 0. Given the astonishing expressive power of deep neural networks, see, e.g., \cite{lin2017does,liang2017deep,rolnick2018power,zhang2021understanding}, they seem likely to produce trivial lower bounds.
\end{itemize}
Overall, generalizing the present work to deep neural networks is highly non-trivial and, at the same, it is unclear if it will provide significantly better results.

In this work, we have shown that Barron's approximation theorem could be used to derive non-trivial lower bounds for the MMSE. However, its implementation is challenging and, when data is post-processed by the additive Gaussian mechanism, it seems to work well only in the large noise regime. While Theorem~\ref{Theorem:2LNNIdentity} could be easily generalized to other families of approximating functions beyond neural networks, it is crucial to find a family with good approximation guarantees for the conditional expectations under consideration. We leave the search for such a family and approximation guarantees as future work.

\section*{Acknowledgments}

The authors would like to thank the anonymous reviewers at ISIT 2021 and ITR3@ICML-21 for their valuable comments on early versions of this work. Also, Mario Diaz would like to thank Hao Wang for useful discussions on early versions of Theorem~\ref{Theorem:CEtaL1NormsBound}.

\appendices
\section{Proof of Theorem~\ref{Theorem:MainTheoremPostProcessing}}
\label{Appendix:MainTheoremPostProcessing}

For each $x\in\mathbb{R}$, we define
\begin{equation}
\label{eq:Defgsigma}
    g^{\sigma}_{\pm}(x) \coloneqq \lambda_{\pm} (f_{\pm} \ast K_{\sigma})(x).
\end{equation}
Observe that, with this notation,
\begin{equation}
    \eta^{\sigma}(x) = \frac{g^{\sigma}_{+}(x) - g^{\sigma}_{-}(x)}{g^{\sigma}_{+}(x) + g^{\sigma}_{-}(x)}.
\end{equation}
The following simple lemma provides useful expressions for the derivatives of $g^{\sigma}_{\pm}$.

\begin{lemma}
\label{Lemma:DerivativesTildeG}
If $j\in\{0,1,2,3\}$, then, for every $x\in\mathbb{R}$,
\begin{equation}
    \frac{\mathrm{d}^{j}g^{\sigma}_{\pm}}{\mathrm{d}x^{j}}(x) = \int_{\mathbb{R}} \lambda_{\pm} f_{\pm}(s) P^{\sigma}_{j}(x-s) K_{\sigma}(x-s) \mathrm{d}s,
\end{equation}
where $P^{\sigma}_{0}(x) = 1$, $P^{\sigma}_{1}(x) = - x/\sigma^{2}$, $P^{\sigma}_{2}(x) = (x^2-\sigma^2)/\sigma^4$ and $P^{\sigma}_{3}(x) = -(x^3-3\sigma^{2}x)/\sigma^{6}$.
\end{lemma}

\begin{proof}
It can be verified that, for each $j\in\{0,1,2,3\}$,
\begin{equation}
    \frac{\mathrm{d}^{j}K_{\sigma}^{(j)}}{\mathrm{d}x^{j}}(x) = P^{\sigma}_{j}(x) K_{\sigma}(x).
\end{equation}
Observe that $g^{\sigma}_{\pm} = (\lambda_{\pm}f_{\pm}) \ast K_{\sigma}$. Therefore, the lemma follows from the general formula $\displaystyle \frac{\mathrm{d}^{j}}{\mathrm{d} x^{j}}(h_{1} \ast h_{2}) = h_1 \ast \frac{\mathrm{d}^{j}}{\mathrm{d} x^{j}} h_{2}$.
\end{proof}

In order to avoid cumbersome notation, we omit the superscript $\sigma$ when there is no risk of confusion, e.g., $g^{\sigma}_{\pm}$ is written as $g_{\pm}$ and $\eta^{\sigma}$ is written as $\eta$. In order to simplify our calculations, we introduce the following notation.

\begin{definition}
For each $n\in\mathbb{N}$, we define
\begin{equation}
\label{eq:DefBigK}
    K_{\delta_{1},\ldots,\delta_{n}}(s_{1},\ldots,s_{n};x) \coloneqq \prod_{i=1}^{n} \lambda_{\delta_{i}} f_{\delta_{i}}(s_i) K_{\sigma}(x-s_{i}),
\end{equation}
where $\delta_{1},\ldots,\delta_{n}\in\{\pm\}$, $s_{1},\ldots,s_{n}\in\mathbb{R}$ and $x\in\mathbb{R}$. Also, for $j_{1},\ldots,j_{n}\in\{0,1,2,3\}$, we define
\begin{equation}
\label{eq:DefBigP}
    P_{j_{1},\ldots,j_{n}}(s_{1},\ldots,s_{n};x) \coloneqq \prod_{i=1}^{n} P_{j_{i}}(x-s_{i}).
\end{equation}
\end{definition}

With the above notation, Lemma~\ref{Lemma:DerivativesTildeG} implies that for every $n\in\mathbb{N}$, $\delta_{1},\ldots,\delta_{n}\in\{\pm\}$, $j_{1},\ldots,j_{n}\in\{0,1,2,3\}$ and $x\in\mathbb{R}$,
\begin{align}
    \label{eq:IntegralRepresentationProducts} \prod_{i=1}^{n} g^{(j_{i})}_{\delta_{i}}(x) = \int_{\mathbb{R}^{n}} P_{j_{1},\ldots,j_{n}}(s;x) K_{\delta_{1},\ldots,\delta_{n}}(s;x) \mathrm{d} s,
\end{align}
where $s = (s_{1},\ldots,s_{n})$ and $\mathrm{d} s = \mathrm{d} s_{1} \cdots \mathrm{d} s_{n}$. In particular, by taking $j_{i} = 0$ for all $i\in[n]$,
\begin{equation}
\label{eq:IntegralRepresentationProducts0}
    \prod_{i=1}^{n} g_{\delta_{i}}(x) = \int_{\mathbb{R}^{n}} K_{\delta_{1},\ldots,\delta_{n}}(s;x) \mathrm{d} s.
\end{equation}
Finally, observe that for every $x\in\mathbb{R}$,
\begin{equation}
\label{eq:IntegralRepresentationProductSupports}
    \mathrm{Supp}(K_{\delta_{1},\ldots,\delta_{n}}(\cdot;x)) = \mathrm{Supp}(f_{\delta_{1}})\times\cdots\times\mathrm{Supp}(f_{\delta_{n}}).
\end{equation}

Now we derive a pointwise bound for $\eta'$.

\begin{lemma}
\label{Lemma:KeyUpperBound1}
If $\mathrm{Supp}(f_{\pm}) \subset [-1,1]$, then, for all $x\in\mathbb{R}$,
\begin{equation}
    \lvert \eta'(x) \rvert \leq \frac{4}{\sigma^2} \frac{g_{+}(x)g_{-}(x)}{\left(g_{+}(x)+g_{-}(x)\right)^{2}}.
\end{equation}
\end{lemma}

\begin{proof}
In Lemma~\ref{Lemma:DerivativesKappa} we prove that, for all $x\in\mathbb{R}$,
\begin{equation}
\label{eq:KeyUpperBound0Kappa}
    \eta'(x) = \frac{\mathrm{I}(x)}{(g_{+}(x) + g_{-}(x))^2},
\end{equation}
where
\begin{equation}
    \mathrm{I}(x) \coloneqq 2\left(g'_{+}(x) g_{-}(x) - g_{+}(x) g'_{-}(x)\right).
\end{equation}
The integral formula in \eqref{eq:IntegralRepresentationProducts} and \eqref{eq:IntegralRepresentationProductSupports} imply that
\begin{equation}
    g'_{+}(x) g_{-}(x) = \int_{-1}^{1} \int_{-1}^{1} P_{1,0}(s;x) K_{+,-}(s;x) \mathrm{d} s.
\end{equation}
\emph{Mutatis mutandis}, we have that
\begin{equation}
    g_{+}(x) g'_{-}(x) = \int_{-1}^{1} \int_{-1}^{1} P_{0,1}(s;x) K_{+,-}(s;x) \mathrm{d} s.
\end{equation}
Thus, we have that
\begin{equation}
\label{eq:KeyUpperBound0Pre}
    \mathrm{I}(x) = \int_{-1}^{1} \int_{-1}^{1} Q(s;x) K_{+,-}(s;x) \mathrm{d} s,
\end{equation}
where $Q(s;x) = 2P_{1,0}(s;x) - 2P_{0,1}(s;x)$. By the definition of $P_{j_{1},\ldots,j_{n}}$ in \eqref{eq:DefBigP},
\begin{align}
    2P_{1,0}(s;x) &= \frac{-2}{\sigma^{2}} x + \frac{2s_{1}}{\sigma^{2}},\\
    - 2P_{0,1}(s;x) &= \frac{2}{\sigma^{2}} x + \frac{-2s_{2}}{\sigma^{2}}.
\end{align}
As a result, $\displaystyle Q(s_{1},s_{2};x) = \frac{2(s_{1}-s_{2})}{\sigma^{2}}$ and \eqref{eq:KeyUpperBound0Pre} becomes
\begin{equation}
\label{eq:KeyUpperBound0Delta}
    \mathrm{I}(x) = \int_{-1}^{1} \int_{-1}^{1} \frac{2(s_{1}-s_{2})}{\sigma^{2}} K_{+,-}(s_{1},s_{2};x) \mathrm{d} s_{1} \mathrm{d} s_{2}.
\end{equation}
Since $|s_{1} - s_{2}| \leq 2$ whenever $s_{1},s_{2}\in[-1,1]$, we have that
\begin{align}
    \lvert \mathrm{I}(x) \rvert &\leq \frac{4}{\sigma^{2}} \int_{-1}^{1} \int_{-1}^{1} K_{+,-}(s_{1},s_{2};x) \mathrm{d} s_{1} \mathrm{d} s_{2}\\
    \label{eq:KeyUpperBound0} &= \frac{4}{\sigma^2} g_{+}(x) g_{-}(x),
\end{align}
where the equality follows from \eqref{eq:IntegralRepresentationProducts0}. The lemma follows by plugging the previous inequality in \eqref{eq:KeyUpperBound0Kappa}.
\end{proof}

\begin{figure*}[t]
\normalsize
\begin{align}
    \label{eq:LongEq1} P_{3,0}(s_{1},s_{2};x) &= \frac{-x^{3} + 3s_{1}x^{2} + 3(\sigma^{2}-s_{1}^{2})x + s_{1}(s_{1}^{2} - 3\sigma^{2})}{\sigma^{6}}\\
    P_{2,1}(s_{1},s_{2};t) &= \frac{-x^{3} + (2s_{1}+s_{2})x^{2} + (\sigma^{2} - s_{1}^{2} - 2s_{1}s_{2})x + s_{2}(s_{1}^{2}-\sigma^{2})}{\sigma^{6}}\\
    -P_{1,2}(s_{1},s_{2};t) &= \frac{x^{3} -(s_{1}+2s_{2})x^{2} - (\sigma^{2} - 2s_{1}s_{2} - s_{2}^{2})x - s_{1}(s_{2}^{2}-\sigma^{2})}{\sigma^{6}}\\
    \label{eq:LongEq4} -P_{0,3}(s_{1},s_{2};x) &= \frac{x^{3} - 3s_{2}x^{2} - 3(\sigma^{2}-s_{2}^{2})x - s_{2}(s_{2}^{2} - 3\sigma^{2})}{\sigma^{6}}
\end{align}
\hrulefill
\end{figure*}

Now we establish a similar upper bound for $\eta''$.

\begin{lemma}
\label{Lemma:KeyUpperBound2}
If $\mathrm{Supp}(f_{\pm}) \subset [-1,1]$, then, for all $x\in\mathbb{R}$,
\begin{equation}
    \lvert \eta''(x) \rvert \leq \frac{8}{\sigma^{4}} \frac{g_{+}(x)g_{-}(x)}{\left(g_{+}(x)+g_{-}(x)\right)^{2}}.
\end{equation}
\end{lemma}

\begin{proof}
In Lemma~\ref{Lemma:DerivativesKappa} we prove that\footnote{Recall that, by Lemma~\ref{Lemma:DerivativesKappa}, $\displaystyle \eta' = 2 \frac{g_{+}'g_{-} - g_{+}g_{-}'}{(g_{+} + g_{-})^2}$.}
\begin{equation}
    \eta''(x) = 2 \frac{g_{+}''g_{-} - g_{+}g_{-}''}{(g_{+} + g_{-})^2} - 4\frac{g_{+}' g_{-} - g_{+} g_{-}'}{(g_{+} + g_{-})^2} \frac{g_{+}' + g_{-}'}{g_{+} + g_{-}}.
\end{equation}
In particular, we have that $\displaystyle \eta''(x) = \frac{\mathrm{I}_{1}(x) + \mathrm{I}_{2}(x)}{(g_{+}(x) + g_{-}(x))^{3}}$ where
\begin{align}
    \mathrm{I}_{1} &= 2g_{+}g_{+}''g_{-} - 2g_{+}g_{+}g_{-}'' - 4g_{+}'g_{+}'g_{-} + 4g_{+}'g_{+}g_{-}',\\
    \mathrm{I}_{2} &= 2g_{+}''g_{-}g_{-} - 2g_{+}g_{-}g_{-}'' - 4g_{+}'g_{-}'g_{-} + 4g_{+}g_{-}'g_{-}'.
\end{align}
The integral formula in \eqref{eq:IntegralRepresentationProducts} implies that, for all $x\in\mathbb{R}$,
\begin{equation}
\label{eq:KeyUpperBound1Pre}
    \mathrm{I}_{1}(x) = \int_{-1}^{1} \int_{-1}^{1} \int_{-1}^{1} Q_{1}(s;x) K_{+,+,-}(s;x) \mathrm{d} s,
\end{equation}
where $Q_{1} = 2P_{0,2,0} - 2 P_{0,0,2} - 4 P_{1,1,0} + 4P_{1,0,1}$. By the definition of $P_{j_{1},\ldots,j_{n}}$ in \eqref{eq:DefBigP},
\begin{align}
    2P_{0,2,0}(s;x) &= \frac{2}{\sigma^{4}} x^{2} + \frac{-4s_{2}}{\sigma^{4}} x + \frac{2(s_{2}^{2}-\sigma^{2})}{\sigma^{4}},\\
    - 2 P_{0,0,2}(s;x) &= \frac{-2}{\sigma^{4}} x^{2} + \frac{4s_{3}}{\sigma^{4}} x + \frac{2(\sigma^{2}-s_{3}^{2})}{\sigma^{4}},\\
    - 4 P_{1,1,0}(s;x) &= \frac{-4}{\sigma^{4}} x^{2} + \frac{4(s_{1}+s_{2})}{\sigma^{4}} x + \frac{-4s_{1}s_{2}}{\sigma^{4}},\\
    4P_{1,0,1}(s;x) &= \frac{4}{\sigma^{4}} x^{2} + \frac{-4(s_{1}+s_{3})}{\sigma^{4}} x + \frac{4s_{1}s_{3}}{\sigma^{4}}.
\end{align}
As a result, we obtain that
\begin{equation}
    Q_{1}(s;x) = \frac{2(s_{2}^{2}-s_{3}^{2}-2s_{1}s_{2}+2s_{1}s_{3})}{\sigma^{4}}.
\end{equation}
The inequality $2s_{1}s_{3} \leq s_{1}^{2} + s_{3}^{2}$ implies that
\begin{equation}
    s_{2}^{2}-s_{3}^{2}-2s_{1}s_{2}+2s_{1}s_{3} \leq (s_{1}-s_{2})^{2}.
\end{equation}
Similarly, the inequality $2s_{1}s_{2} \leq s_{1}^{2} + s_{2}^{2}$ implies that
\begin{equation}
    s_{2}^{2}-s_{3}^{2}-2s_{1}s_{2}+2s_{1}s_{3} \geq -(s_{1}-s_{3})^{2}.
\end{equation}
In particular, $\lvert \displaystyle Q_{1}(s;x) \rvert \leq \frac{8}{\sigma^{4}}$ whenever $s_{1},s_{2},s_{3}\in[-1,1]$. Therefore, \eqref{eq:KeyUpperBound1Pre} implies that
\begin{align}
    |\mathrm{I}_{1}(x)| &\leq \frac{8}{\sigma^{4}} \int_{-1}^{1} \int_{-1}^{1} \int_{-1}^{1} K_{+,+,-}(s;x) \mathrm{d} s\\
    \label{eq:KeyUpperBound1PartI} &= \frac{8}{\sigma^{4}} g_{+}(x)g_{+}(x)g_{-}(x),
\end{align}
where the equality follows from \eqref{eq:IntegralRepresentationProducts0}. \emph{Mutatis mutandis}, it can be shown that
\begin{equation}
\label{eq:KeyUpperBound1Pre2}
    \mathrm{I}_{2}(x) = \int_{-1}^{1} \int_{-1}^{1} \int_{-1}^{1} Q_{2}(s;x) K_{+,-,-}(s;x) \mathrm{d} s,
\end{equation}
where
\begin{equation}
\label{eq:KeyUpperBound1Q2}
    Q_{2}(s;x) = \frac{2(s_{1}^{2}-s_{3}^{2}-2s_{1}s_{2}+2s_{2}s_{3})}{\sigma^{4}}.
\end{equation}
As before, \eqref{eq:KeyUpperBound1Pre2} and \eqref{eq:KeyUpperBound1Q2} imply that, for all $x\in\mathbb{R}$,
\begin{equation}
\label{eq:KeyUpperBound1PartII}
    |\mathrm{I}_{2}(x)| \leq \frac{8}{\sigma^{4}} g_{+}(x)g_{-}(x)g_{-}(x).
\end{equation}
Since $\displaystyle \eta''(x) = \frac{\mathrm{I}_{1}(x) + \mathrm{I}_{2}(x)}{(g_{+}(x) + g_{-}(x))^{3}}$, \eqref{eq:KeyUpperBound1PartI} and \eqref{eq:KeyUpperBound1PartII} imply that
\begin{equation}
    |\eta''(x)| \leq \frac{8}{\sigma^{4}} \frac{g_{+}(x)g_{-}(x)}{\left(g_{+}(x)+g_{-}(x)\right)^{2}},
\end{equation}
as required.
\end{proof}

Finally, we establish an upper bound for $\eta'''$ akin to those in the previous lemmas.

\begin{lemma}
\label{Lemma:KeyUpperBound3}
If $\mathrm{Supp}(f_{\pm}) \subset [-1,1]$, then, for all $x\in\mathbb{R}$,
\begin{equation}
\label{eq:KeyUpperBound3}
    \lvert\eta'''(x)\rvert \leq \frac{16x^{2} + 44\lvert x \rvert + 34+12\sigma^{2}}{\sigma^{6}} \frac{g_{+}(x)g_{-}(x)}{\left(g_{+}(x)+g_{-}(x)\right)^{2}}.
\end{equation}
\end{lemma}

\begin{proof}
In Lemma~\ref{Lemma:DerivativesKappa} we prove that
\begin{equation}
\label{eq:KeyUpperBoundEtaDecomposition}
    \eta''' = \frac{\mathrm{I}}{(g_{+} + g_{-})^2} - 2 \eta' \frac{g_{+}'' + g_{-}''}{g_{+} + g_{-}} - 3 \eta'' \frac{g_{+}' + g_{-}'}{g_{+} + g_{-}},
\end{equation}
where
\begin{equation}
    I = 2 (g_{+}'''g_{-} + g_{+}''g_{-}' - g_{+}'g_{-}'' - g_{+}g_{-}''').
\end{equation}
The integral formula in \eqref{eq:IntegralRepresentationProducts} implies that, for all $x\in\mathbb{R}$,
\begin{equation}
\label{eq:KeyUpperBound3Pre}
    \mathrm{I}(x) = \int_{-1}^{1} \int_{-1}^{1} \int_{-1}^{1} Q(s;x) K_{+,+,-}(s;x) \mathrm{d} s,
\end{equation}
where $Q = P_{3,0} + P_{2,1} - P_{1,2} - P_{0,3}$. By the definition of $P_{j_{1},\ldots,j_{n}}$ in \eqref{eq:DefBigP} and equations~\eqref{eq:LongEq1} -- \eqref{eq:LongEq4}, we conclude that, for all $x\in\mathbb{R}$,
\begin{align}
    Q(s_{1},s_{2};x) &= \frac{4(s_{1}-s_{2})x^{2} + 4(s_{2}^{2}-s_{1}^{2})x}{\sigma^{6}} + \frac{(s_{1}-s_{2})(s_{1}+s_{2})^{2}-2(s_{1}-s_{2})\sigma^{2}}{\sigma^{6}}.
\end{align}
In particular, for $s_{1},s_{2}\in[-1,1]$,
\begin{equation}
    \lvert Q(s;x) \rvert \leq \frac{8x^{2}+4\lvert x \rvert + 2 + 4\sigma^{2}}{\sigma^{6}}.
\end{equation}
Therefore, \eqref{eq:KeyUpperBound3Pre} implies that
\begin{align}
    \lvert \mathrm{I}(x) \rvert &\leq \frac{8x^{2}+4\lvert x \rvert + 2 + 4\sigma^{2}}{\sigma^{6}} \int_{-1}^{1} \int_{-1}^{1} K_{+,-}(s;t) \mathrm{d}s\\
    \label{eq:KeyUpperBound3I} &= \frac{8x^{2}+4\lvert x \rvert + 2 + 4\sigma^{2}}{\sigma^{6}} g_{+}(x) g_{-}(x),
\end{align}
where the equality follows from \eqref{eq:IntegralRepresentationProducts0}.

Lemma~\ref{Lemma:DerivativesTildeG} and the fact that $\mathrm{Supp}(f_{\pm}) \subset [-1,1]$ imply that
\begin{equation}
    \lvert g_{\pm}''(x) \rvert \leq \frac{x^{2}+2\lvert x \rvert + 1 + \sigma^{2}}{\sigma^{4}} g_{\pm}(x).
\end{equation}
As a result, we obtain that
\begin{equation}
    \left\lvert 2 \eta'(x) \frac{g_{+}''(x) + g_{-}''(x)}{g_{+}(x) + g_{-}(x)} \right\rvert \leq \frac{2x^{2}+4\lvert x \rvert + 2+2\sigma^{2}}{\sigma^{4}} \lvert \eta'(x) \rvert.
\end{equation}
Therefore, by Lemma~\ref{Lemma:KeyUpperBound1},
\begin{align}
    \label{eq:KeyUpperBound3II} \left\lvert 2 \eta'(x) \frac{g_{+}''(x) + g_{-}''(x)}{g_{+}(x) + g_{-}(x)} \right\rvert &\leq \frac{8x^{2}+16\lvert x \rvert + 8+8\sigma^{2}}{\sigma^{6}} \frac{g_{+}(x)g_{-}(x)}{\left(g_{+}(x)+g_{-}(x)\right)^{2}}.
\end{align}

Lemma~\ref{Lemma:DerivativesTildeG} and the fact that $\mathrm{Supp}(f_{\pm}) \subset [-1,1]$ imply that
\begin{equation}
    \lvert g_{\pm}'(x)\rvert \leq \frac{\lvert x \rvert + 1}{\sigma^{2}} g_{\pm}(x).
\end{equation}
Therefore, we conclude that
\begin{align}
    \left\lvert 3\eta''(x) \frac{g_{+}'(x) + g_{-}'(x)}{g_{+}(x) + g_{-}(x)} \right\rvert &\leq \frac{3(\lvert x \rvert + 1)}{\sigma^{2}} \lvert \eta''(x) \rvert\\
    \label{eq:KeyUpperBound3III} &\leq \frac{24(\lvert x \rvert + 1)}{\sigma^{6}} \frac{g_{+}(x)g_{-}(x)}{\left(g_{+}(x)+g_{-}(x)\right)^{2}},
\end{align}
where the last inequality follows from Lemma~\ref{Lemma:KeyUpperBound2}. By plugging \eqref{eq:KeyUpperBound3I}, \eqref{eq:KeyUpperBound3II} and \eqref{eq:KeyUpperBound3III} in \eqref{eq:KeyUpperBoundEtaDecomposition}, the result follows.
\end{proof}

It is possible to obtain sightly better constants than those in \eqref{eq:KeyUpperBound3} by avoiding the use of Lemmas~\ref{Lemma:KeyUpperBound1} and \ref{Lemma:KeyUpperBound2}. However, the complexity of the proof increases considerably and the benefit is marginal given that the $L^{1}$-norm of $\eta'''$ only appears inside a logarithm.

\begin{proof}[\textbf{Proof of Theorem~\ref{Theorem:MainTheoremPostProcessing}}]
By Theorem~\ref{Theorem:CEtaL1NormsBound}, we have that
\begin{equation}
\label{eq:ThmAdditiveMechanismProofCEta0}
    C_{\eta} \leq \frac{2\sqrt{2}}{\sqrt{\pi}} \left(1 + \frac{1}{2}\log\left(\lVert \eta'\rVert_1 \lVert \eta'''\rVert_1\right) - \log(\lVert \eta'' \rVert_{1})\right) \lVert \eta''\rVert_{1}.
\end{equation}
Recall the definition of $M_{\alpha}$ in \eqref{eq:DefMomentsMp}. By Lemmas~\ref{Lemma:KeyUpperBound1} -- \ref{Lemma:KeyUpperBound3}, we have that
\begin{align}
    \lVert \eta'\rVert_1 &\leq \frac{4M_{0}}{\sigma^{2}},\\
    \lVert \eta''\rVert_1 &\leq \frac{8M_{0}}{\sigma^{4}},\\
    \lVert \eta'''\rVert_1 &\leq \frac{16M_{2} + 44M_{1} + (34+12\sigma^{2})M_{0}}{\sigma^{6}}.
\end{align}
As a result, for all $\sigma>0$,
\begin{align}
    \label{eq:ThmAdditiveMechanismProofCEta1} C_{\eta} &\leq \frac{16\sqrt{2}M_{0}}{\sqrt{\pi}\sigma^{4}} \left(1 + \frac{1}{2} \log\left(\frac{M}{\sigma^{8}}\right)\right) - \frac{2\sqrt{2}}{\sqrt{\pi}} \log(\lVert \eta'' \rVert_{1}) \lVert \eta''\rVert_{1},
\end{align}
where
\begin{equation}
    M \coloneqq 64M_{2}M_{0} + 176M_{1}M_{0} + (136 + 48\sigma^{2})M_{0}^{2}.
\end{equation}
Since $-\log(z)z \leq 1/e$ for all $z\in[0,\infty)$, \eqref{eq:ThmAdditiveMechanismProofCEta0} implies that, for all $\sigma>0$,
\begin{equation}
    C_{\eta} \leq \frac{2\sqrt{2}}{e\sqrt{\pi}} + \frac{16\sqrt{2}M_{0}}{\sqrt{\pi}\sigma^{4}} \left(1 + \frac{1}{2} \log\left(\frac{M}{\sigma^{8}}\right)\right).
\end{equation}
It is straightforward to verify that $z \mapsto -\log(z)z$ is increasing over $[0,1/e]$. Thus, if $\sqrt[4]{8eM_{0}} \leq \sigma$, \eqref{eq:ThmAdditiveMechanismProofCEta1} implies that
\begin{equation}
    C_{\eta} \leq \frac{16\sqrt{2}M_{0}}{\sqrt{\pi}\sigma^{4}} \left(1 + \frac{1}{2} \log\left(\frac{M}{64M_{0}^{2}}\right)\right).
\end{equation}
After some manipulations, \eqref{eq:MainTheoremPostProcessingLargeSigma} follows.
\end{proof}

\section{Proof of Proposition~\ref{Proposition:MpSeparated}}
\label{Appendix:ProofMPSeparated}

Recall that, for each $x\in\mathbb{R}$,
\begin{equation}
    g^{\sigma}_{\pm}(x) \coloneqq \lambda_{\pm} (f_{\pm} \ast K_{\sigma})(x).
\end{equation}
The next lemma provides upper and lower bounds for $g^{\sigma}_{\pm}(x)$ under the assumptions of Proposition~\ref{Proposition:MpSeparated}.

\begin{lemma}
\label{Lemma:BoundsTildegpmSepparated}
In the context of Proposition~\ref{Proposition:MpSeparated}, for all $x\in\mathbb{R}$,
\begin{align*}
    & \lambda_{+} \min_{s\in[x-1,x-\gamma]} K_{\sigma}(s) \leq g^{\sigma}_{+}(x) \leq \lambda_{+} \max_{s\in[x-1,x-\gamma]} K_{\sigma}(s),\\
    & \lambda_{-} \min_{s\in[x+\gamma,x+1]} K_{\sigma}(s) \leq g^{\sigma}_{-}(x) \leq \lambda_{-} \max_{s\in[x+\gamma,x+1]} K_{\sigma}(s).
\end{align*}
\end{lemma}

\begin{proof}
By assumption $\mathrm{Supp}(f_{+}) \subset [\gamma,1]$, thus
\begin{equation}
    g^{\sigma}_{+}(x) = \lambda_{+} \int_{\gamma}^{1} f_{+}(s) K_{\sigma}(x-s) \mathrm{d}s.
\end{equation}
Therefore, for all $x\in\mathbb{R}$,
\begin{align}
    \lambda_{+} \left(\min_{s\in[\gamma,1]} K_{\sigma}(x-s)\right) \int_{\gamma}^{r} f_{+}(s) \mathrm{d}s \leq g^{\sigma}_{+}(x) \leq \lambda_{+} \left(\max_{s\in[\gamma,1]} K_{\sigma}(x-s)\right) \int_{\gamma}^{1} f_{+}(s) \mathrm{d}s.
\end{align}
Since $\displaystyle \int_{\gamma}^{r} f_{+}(s) \mathrm{d}s = 1$, the inequalities for $g_{+}^{\sigma}$ follow. The inequalities for $g_{-}^{\sigma}$ are proved \emph{mutatis mutandis}.
\end{proof}

The next lemma provides an upper bound for $g^{\sigma}_{\pm}/(g^{\sigma}_{+} + g^{\sigma}_{-})$.

\begin{lemma}
\label{Lemma:UpperBoundQuotientSepparated}
In the context of Proposition~\ref{Proposition:MpSeparated},
\begin{itemize}
    \item $\displaystyle \frac{g^{\sigma}_{+}(x)}{g^{\sigma}_{+}(x) + g^{\sigma}_{-}(x)} \leq \frac{\lambda_{+}}{\lambda_{-}} e^{-2\gamma \lvert x \rvert/\sigma^{2}}$, $x < -1$;
    \item $\displaystyle \frac{g^{\sigma}_{-}(x)}{g^{\sigma}_{+}(x) + g^{\sigma}_{-}(x)} \leq \frac{\lambda_{-}}{\lambda_{+}} e^{-2\gamma \vert x \rvert/\sigma^2}$, $x > 1$.
\end{itemize}
\end{lemma}

\begin{proof}
By Lemma~\ref{Lemma:BoundsTildegpmSepparated}, for all $x\in\mathbb{R}$,
\begin{equation}
    \frac{g^{\sigma}_{+}(x)}{g^{\sigma}_{+}(x) + g^{\sigma}_{-}(x)} \leq \frac{g^{\sigma}_{+}(x)}{g^{\sigma}_{-}(x)} \leq \frac{\lambda_{+} \max_{s\in[x-1,x-\gamma]} K_{\sigma}(s)}{\lambda_{-} \min_{s\in[x+\gamma,x+1]} K_{\sigma}(s)}.
\end{equation}
When $x < -1$, we have that
\begin{equation}
    [x-1,x-\gamma],[x+\gamma,x+1]\subset(-\infty,0).
\end{equation}
Since $K_{\sigma}$ is increasing on $(-\infty,0)$,
\begin{equation}
    \frac{g^{\sigma}_{+}(x)}{g^{\sigma}_{+}(x) + g^{\sigma}_{-}(x)} \leq \frac{\lambda_{+}}{\lambda_{-}} \frac{K_{\sigma}(x-\gamma)}{K_{\sigma}(x+\gamma)} = \frac{\lambda_{+}}{\lambda_{-}} e^{-2\gamma \lvert x \rvert/\sigma^{2}}.
\end{equation}
\emph{Mutatis mutandis}, it can be shown that, for $x > 1$, we have the inequality $\displaystyle \frac{g^{\sigma}_{-}(x)}{g^{\sigma}_{+}(x) + g^{\sigma}_{-}(x)} \leq \frac{\lambda_{-}}{\lambda_{+}} e^{-2\gamma \lvert x \rvert/\sigma^2}$.
\end{proof}

Now we are in position to prove Proposition~\ref{Proposition:MpSeparated}.

\begin{proof}[\textbf{Proof of Proposition~\ref{Proposition:MpSeparated}}]
By definition, we have that
\begin{equation}
    M_{\alpha}^{\sigma} = \left(\int_{-\infty}^{-1} + \int_{-1}^{1} + \int_{1}^{\infty}\right) \lvert x \rvert^{p} \frac{g^{\sigma}_{+}(x)g^{\sigma}_{-}(x)}{\left(g^{\sigma}_{+}(x)+g^{\sigma}_{-}(x)\right)^{2}} \mathrm{d} x.
\end{equation}
Lemma~\ref{Lemma:UpperBoundQuotientSepparated} and the trivial upper bound $\displaystyle \frac{g^{\sigma}_{\pm}(x)}{g^{\sigma}_{+}(x)+g^{\sigma}_{-}(x)} \leq 1$ imply that
\begin{equation}
    M_{\alpha}^{\sigma} \leq 2 + \frac{\lambda_{+}^{2}+\lambda_{-}^{2}}{\lambda_{+}\lambda_{-}} \int_{1}^{\infty} x^{p} \exp\left\{-2\gamma x/\sigma^{2}\right\} \mathrm{d} x.
\end{equation}
By the formulas,
\begin{align}
    \int e^{-\beta x} \mathrm{d}x &= - \frac{e^{-\beta x}}{\beta},\\
    \int x e^{-\beta x} \mathrm{d}x &= - \frac{e^{-\beta x}}{\beta^{2}}(\beta x + 1),\\
    \int x^{2} e^{-\beta x} \mathrm{d}x &= - \frac{e^{-\beta x}}{\beta^{3}}(\beta^{2}x^{2} + 2\beta x + 2),
\end{align}
the result follows.
\end{proof}

\section{Proof of Proposition~\ref{Proposition:MpOverlapping}}
\label{Appendix:ProofMPOverlapping}

Recall that, for each $x\in\mathbb{R}$, we define
\begin{equation}
    g^{\sigma}_{\pm}(x) := \lambda_{\pm} (f_{\pm} \ast K_{\sigma})(x).
\end{equation}
The next lemma provides upper and lower bounds for $g_{\pm}(x)$ under the assumptions of Proposition~\ref{Proposition:MpOverlapping}.

\begin{lemma}
\label{Lemma:BoundsTildegpmOverlapping}
In the context of Proposition~\ref{Proposition:MpOverlapping}, for all $x\in\mathbb{R}$,
\begin{align*}
    & \lambda_{+} \delta_{+} \min_{s\in[x-1,x-\gamma]} K_{\sigma}(s) \leq g^{\sigma}_{+}(x) \leq \lambda_{+} \max_{s\in[x-1,x+\gamma_{0}]} K_{\sigma}(s),\\
    & \lambda_{-} \delta_{-} \min_{s\in[x+\gamma,x+1]} K_{\sigma}(s) \leq g^{\sigma}_{-}(x) \leq \lambda_{-} \max_{s\in[x-\gamma_{0},x+1]} K_{\sigma}(s).
\end{align*}
\end{lemma}

\begin{proof}
By assumption $\mathrm{Supp}(f_{+}) \subset [-\gamma_{0},1]$, thus
\begin{equation}
    g^{\sigma}_{+}(x) = \lambda_{+} \int_{-\gamma_{0}}^{1} f_{+}(s) K_{\sigma}(x-s) {\rm d}s.
\end{equation}
Since $-\gamma_{0} < \gamma < 1$, for all $x\in\mathbb{R}$,
\begin{align}
    g^{\sigma}_{+}(x) &\geq \lambda_{+} \int_{\gamma}^{1} f_{+}(s) K_{\sigma}(x-s) \textrm{d} s\\
    &\geq \lambda_{+} \int_{\gamma}^{1} f_{+}(s) \textrm{d} s \min_{s\in[\gamma,1]} K_{\sigma}(x-s)\\
    &= \lambda_{+} \delta_{+} \min_{s\in[x-1,x-\gamma]} K_{\sigma}(s).
\end{align}
Similarly, for all $x\in\mathbb{R}$,
\begin{align}
    g^{\sigma}_{+}(x) &\leq \lambda_{+} \int_{-\gamma_{0}}^{1} f_{+}(s) \textrm{d} s \max_{s\in[-\gamma_{0},1]} K_{\sigma}(x-s)\\
    &= \lambda_{+} \max_{s\in[t-1,t+\gamma_{0}]} K_{\sigma}(s).
\end{align}
The inequalities for $g^{\sigma}_{-}$ are proved \emph{mutatis mutandis}.
\end{proof}

The next lemma provides an upper bound for $g^{\sigma}_{\pm}/(g^{\sigma}_{+} + g^{\sigma}_{-})$.

\begin{lemma}
\label{Lemma:UpperBoundQuotientOverlapping}
In the context of Proposition~\ref{Proposition:MpOverlapping},
\begin{itemize}
    \item $\displaystyle \frac{g^{\sigma}_{+}(x)}{g^{\sigma}_{+}(x) + g^{\sigma}_{-}(x)} \leq \frac{\lambda_{+}}{\delta_{-}\lambda_{-}} e^{-(\gamma-\gamma_{0})(\lvert x \rvert - 1)/\sigma^{2}}$, $x < -1$;
    \item $\displaystyle \frac{g^{\sigma}_{-}(x)}{g^{\sigma}_{+}(x) + g^{\sigma}_{-}(x)} \leq \frac{\lambda_{-}}{\delta_{+}\lambda_{+}} e^{-(\gamma-\gamma_{0})(\lvert x \rvert - 1)/\sigma^{2}}$, $x > 1$.
\end{itemize}
\end{lemma}

\begin{proof}
By Lemma~\ref{Lemma:BoundsTildegpmOverlapping}, for all $x\in\mathbb{R}$,
\begin{equation}
    \frac{g^{\sigma}_{+}(x)}{g^{\sigma}_{+}(x) + g^{\sigma}_{-}(x)} \leq \frac{g^{\sigma}_{+}(x)}{g^{\sigma}_{-}(x)} \leq \frac{\lambda_{+} \max_{s\in[x-1,x+\gamma_{0}]} K_{\sigma}(s)}{\delta_{-}\lambda_{-} \min_{s\in[x+\gamma,x+1]} K_{\sigma}(s)}.
\end{equation}
When $x < -1$, we have that
\begin{equation}
    [x-1,x+\gamma_{0}],[x+\gamma,x+1]\subset(-\infty,0).
\end{equation}
Since $K_{\sigma}$ is increasing on $(-\infty,0)$,
\begin{align}
    \frac{g^{\sigma}_{+}(x)}{g^{\sigma}_{+}(x) + g^{\sigma}_{-}(x)} &\leq \frac{\lambda_{+}}{\delta_{-}\lambda_{-}} \frac{K_{\sigma}(x+\gamma_{0})}{K_{\sigma}(x+\gamma)}\\
    &= \frac{\lambda_{+}}{\delta_{-}\lambda_{-}} e^{-(\gamma-\gamma_{0})(\lvert x \rvert - 1)/\sigma^{2}}.
\end{align}
\emph{Mutatis mutandis}, it can be shown that, for $x>1$, we have the inequality $\displaystyle \frac{g^{\sigma}_{-}(x)}{g^{\sigma}_{+}(x) + g^{\sigma}_{-}(x)} \leq \frac{\lambda_{-}}{\delta_{+}\lambda_{+}} e^{-(\gamma-\gamma_{0})(\lvert x \rvert - 1)/\sigma^{2}}$.
\end{proof}

Now we are in position to prove Proposition~\ref{Proposition:MpOverlapping}.

\begin{proof}[\textbf{Proof of Proposition~\ref{Proposition:MpOverlapping}}]
By definition, we have that
\begin{equation}
    M^{\sigma}_{p} = \left(\int_{-\infty}^{-1} + \int_{-1}^{1} + \int_{1}^{\infty}\right) \lvert x \rvert^{p} \frac{g^{\sigma}_{+}(x)g^{\sigma}_{-}(x)}{(g^{\sigma}_{+}(x)+g^{\sigma}_{-}(x))^{2}}{\rm d} x.
\end{equation}
Lemma~\ref{Lemma:UpperBoundQuotientOverlapping} and the trivial upper bound $\displaystyle \frac{g^{\sigma}_{\pm}(x)}{g^{\sigma}_{+}(x)+g^{\sigma}_{-}(x)} \leq 1$ imply that
\begin{equation}
    M^{\sigma}_{p} \leq 2 + \frac{\delta_{+}\lambda_{+}^{2}+\delta_{-}\lambda_{-}^{2}}{\delta_{+}\lambda_{+}\delta_{-}\lambda_{-}} \int_{1}^{\infty} x^{p} e^{-(\gamma-\gamma_{0})(x - 1)/\sigma^{2}} {\rm d} x.
\end{equation}
By the formulas,
\begin{align}
    \int e^{-\beta x} {\rm d}x &= - \frac{e^{-\beta x}}{\beta},\\
    \int x e^{-\beta x} {\rm d}x &= - \frac{e^{-\beta x}}{\beta^{2}}(\beta x + 1),\\
    \int x^{2} e^{-\beta x} {\rm d}x &= - \frac{e^{-\beta x}}{\beta^{3}}(\beta^{2}x^{2} + 2\beta x + 2),
\end{align}
the result follows.
\end{proof}

\section{Proof of Theorem~\ref{Theorem:MainTheoremPostProcessingTheta}}
\label{Appendix:MainTheoremPostProcessingTheta}

Let $p_{+} \coloneqq p$ and $p_{-} \coloneqq 1 - p$. For each $x\in\mathbb{R}$, we define
\begin{equation}
\label{eq:Defgsigmatheta}
    g^{\sigma}_{\pm}(x) \coloneqq p_{\pm}(f_{\pm} \ast K_{\sigma})(x) + \lambda_{\pm}.
\end{equation}
Observe that, with this notation,
\begin{equation}
\label{eq:AppendixDefTheta}
    \theta^{\sigma}(x) = \frac{1}{2}\log\left(\frac{g^{\sigma}_{+}(x)}{g^{\sigma}_{-}(x)}\right).
\end{equation}
The following lemma provides useful expressions for the $L^{1}$, $L^{2}$ and $L^{3}$-norms of the first derivative of $K_{\sigma}$.

\begin{lemma}
\label{Lemma:NormDerivativesKSigma1}
If $\displaystyle K_{\sigma}(x) = \frac{e^{-x^{2}/2\sigma^{2}}}{\sqrt{2\pi\sigma^{2}}}$, then $\displaystyle \lVert K_{\sigma}' \rVert_{1} = \frac{\sqrt{2}}{\sqrt{\pi}\sigma}$, $\displaystyle \lVert K_{\sigma}' \rVert_{2}^{2} = \frac{1}{4\sqrt{\pi}\sigma^{3}}$ and $\displaystyle \lVert K_{\sigma}' \rVert_{3}^{3} = \frac{\sqrt{2}}{9\sqrt{\pi^{3}}\sigma^{5}}$.
\end{lemma}

\begin{proof}
It can be verified that, for all $x\in\mathbb{R}$,
\begin{equation}
    K_{\sigma}'(x) = -\frac{x}{\sigma^{2}} K_{\sigma}(x).
\end{equation}
Thus, we have that
\begin{equation}
    \lVert K_{\sigma}' \rVert_{1} = \frac{1}{\sigma^{2}} \int_{\mathbb{R}} \lvert x \rvert \frac{e^{-x^{2}/2\sigma^{2}}}{\sqrt{2\pi\sigma^{2}}} \mathrm{d}x.
\end{equation}
Note that the previous integral is the first absolute moment of a Gaussian random variable with mean 0 and variance $\sigma^{2}$. Therefore, we obtain that
\begin{equation}
    \lVert K_{\sigma}' \rVert_{1} = \frac{\sqrt{2}}{\sqrt{\pi}\sigma}.
\end{equation}
Similarly, we have that
\begin{equation}
    \lVert K_{\sigma}' \rVert_{2}^{2} = \frac{1}{2\sqrt{\pi}\sigma^{5}} \int_{\mathbb{R}} x^{2} \frac{e^{-x^{2}/\sigma^{2}}}{\sqrt{\pi\sigma^{2}}} \mathrm{d} x.
\end{equation}
Note that the previous integral is the second moment of a Gaussian random variable with mean 0 and variance $\sigma^{2}/2$. Therefore, we obtain that
\begin{equation}
    \lVert K_{\sigma}' \rVert_{2}^{2} = \frac{1}{4\sqrt{\pi}\sigma^{3}}.
\end{equation}
Finally, we have that
\begin{equation}
    \lVert K_{\sigma}' \rVert_{3}^{3} = \frac{1}{2\sqrt{3}\pi\sigma^{8}} \int_{\mathbb{R}} \lvert x \rvert^{3} \frac{e^{-3x^{2}/2\sigma^{2}}}{\sqrt{2\pi\sigma^{2}/3}} \mathrm{d} x.
\end{equation}
Note that the previous integral is the third absolute moment of a Gaussian random variable with mean 0 and variance $\sigma^{2}/3$. Therefore, we obtain that
\begin{equation}
    \lVert K_{\sigma}' \rVert_{3}^{3} = \frac{\sqrt{2}}{9\sqrt{\pi^{3}}\sigma^{5}},
\end{equation}
as required.
\end{proof}

The following lemma provides useful expressions for the $L^{1}$ and $L^{2}$-norms of the second derivative of $K_{\sigma}$.

\begin{lemma}
\label{Lemma:NormDerivativesKSigma2}
If $\displaystyle K_{\sigma}(x) = \frac{e^{-x^{2}/2\sigma^{2}}}{\sqrt{2\pi\sigma^{2}}}$, then $\displaystyle \lVert K_{\sigma}'' \rVert_{1} \leq \frac{2}{\sigma^{2}}$ and $\displaystyle \lVert K_{\sigma}'' \rVert_{2} = \frac{\sqrt{3}}{2\sqrt{2}\sqrt[4]{\pi}\sigma^{5/2}}$.
\end{lemma}

\begin{proof}
It can be verified that, for all $x\in\mathbb{R}$,
\begin{equation}
    K_{\sigma}''(x) = \frac{x^{2}-\sigma^{2}}{\sigma^{4}} K_{\sigma}(x).
\end{equation}
Thus, we have that
\begin{align}
    \lVert K_{\sigma}'' \rVert_{1} &\leq  \int_{\mathbb{R}} \left(\frac{x^{2}}{\sigma^{4}} + \frac{1}{\sigma^{2}}\right) \frac{e^{-x^{2}/2\sigma^{2}}}{\sqrt{2\pi\sigma^{2}}} \mathrm{d}x\\
    &= \frac{1}{\sigma^{4}} \int_{\mathbb{R}} x^{2} \frac{e^{-x^{2}/2\sigma^{2}}}{\sqrt{2\pi\sigma^{2}}} \mathrm{d}x + \frac{1}{\sigma^{2}}.
\end{align}
Note that the last integral is the second moment of a Gaussian random variable with mean 0 and variance $\sigma^{2}$. Therefore,
\begin{equation}
    \lVert K_{\sigma}'' \rVert_{1} \leq \frac{2}{\sigma^{2}}.
\end{equation}
Similarly, we have that
\begin{equation}
    \lVert K_{\sigma}'' \rVert_{2}^{2} = \int_{\mathbb{R}} \left(\frac{x^{4}}{2\sqrt{\pi}\sigma^{9}} - \frac{x^{2}}{\sqrt{\pi}\sigma^{7}} + \frac{1}{2\sqrt{\pi}\sigma^{5}}\right) \frac{e^{-x^{2}/\sigma^{2}}}{\sqrt{\pi\sigma^{2}}} \mathrm{d} x.
\end{equation}
Note that the last integral is determined by the even moments of a Gaussian random variable with mean 0 and variance $\sigma^{2}/2$. Therefore, we obtain that
\begin{align}
    \lVert K_{\sigma}'' \rVert_{2} &= \left(\frac{3}{8\sqrt{\pi}\sigma^{5}} - \frac{1}{2\sqrt{\pi}\sigma^{5}} + \frac{1}{2\sqrt{\pi}\sigma^{5}}\right)^{1/2}\\
    &= \frac{\sqrt{3}}{2\sqrt{2}\sqrt[4]{\pi}\sigma^{5/2}},
\end{align}
as required.
\end{proof}

The following lemma provides useful expressions for the $L^{1}$-norm of the third derivative of $K_{\sigma}$.

\begin{lemma}
\label{Lemma:NormDerivativesKSigma3}
If $\displaystyle K_{\sigma}(x) = \frac{e^{-x^{2}/2\sigma^{2}}}{\sqrt{2\pi\sigma^{2}}}$, then $\displaystyle \lVert K_{\sigma}''' \rVert_{1} \leq  \frac{5\sqrt{2}}{\sqrt{\pi}\sigma^{3}}$.
\end{lemma}

\begin{proof}
It can be verified that, for all $x\in\mathbb{R}$,
\begin{equation}
    K_{\sigma}'''(x) = -\frac{x^3-3\sigma^{2}x}{\sigma^{6}} K_{\sigma}(x).
\end{equation}
Thus, we have that
\begin{equation}
    \lVert K_{\sigma}''' \rVert_{1} \leq  \int_{\mathbb{R}} \left(\frac{\lvert x \rvert^{3}}{\sigma^{6}} + \frac{3\lvert x \rvert}{\sigma^{4}}\right) \frac{e^{-x^{2}/2\sigma^{2}}}{\sqrt{2\pi\sigma^{2}}} \mathrm{d}x.
\end{equation}
Note that the last integral is determined by the absolute moments of a Gaussian random variable with mean 0 and variance $\sigma^{2}$. Therefore,
\begin{equation}
    \lVert K_{\sigma}''' \rVert_{1} \leq \frac{2\sqrt{2}}{\sqrt{\pi}\sigma^{3}} + \frac{3\sqrt{2}}{\sqrt{\pi}\sigma^{3}} = \frac{5\sqrt{2}}{\sqrt{\pi}\sigma^{3}},
\end{equation}
as required.
\end{proof}

In order to avoid cumbersome notation, we omit the superscript $\sigma$ when there is no risk of confusion, e.g., $g^{\sigma}_{\pm}$ is written as $g_{\pm}$ and $\theta^{\sigma}$ is written as $\theta$. The following corollary provides an upper bound for the $L^{1}$-norm of $\theta'$.

\begin{corollary}
\label{Corollary:L1NormTheta1}
If $f_{\pm}$ is a probability density function, then
\begin{equation}
    \lVert \theta' \rVert_{1} \leq \left(\frac{p_{+}}{\lambda_{+}} + \frac{p_{-}}{\lambda_{-}}\right) \frac{1}{\sqrt{2\pi}\sigma}.
\end{equation}
\end{corollary}

\begin{proof}
In Lemma~\ref{Lemma:DerivativesTheta} we prove that, for all $x\in\mathbb{R}$,
\begin{equation}
    2\theta'(x) = \frac{g_{+}'(x)}{g_{+}(x)} - \frac{g_{-}'(x)}{g_{-}(x)}.
\end{equation}
By the triangle inequality, we have that
\begin{align}
    2 \lvert \theta'(x) \rvert &\leq \frac{\lvert g_{+}'(x) \rvert}{g_{+}(x)} + \frac{\lvert g_{-}'(x) \rvert}{g_{-}(x)}\\
    &\leq \frac{\lvert g_{+}'(x) \rvert}{\lambda_{+}} + \frac{\lvert g_{-}'(x) \rvert}{\lambda_{-}},
\end{align}
where the last inequality follows trivially from \eqref{eq:Defgsigmatheta}. Thus,
\begin{equation}
    \lVert \theta' \rVert_{1} \leq \frac{\lVert g_{+}' \rVert_{1}}{2\lambda_{+}} + \frac{\lVert g_{-}' \rVert_{1}}{2\lambda_{-}}.
\end{equation}
From \eqref{eq:Defgsigmatheta}, it is immediate to see that $g_{+}' = p_{+} (f_{+} \ast K_{\sigma})'$. Hence, the formula $ (h_{1} \ast h_{2})' = h_{1} \ast h_{2}'$ implies that
\begin{equation}
    \lVert \theta' \rVert_{1} \leq \frac{p_{+} \lVert f_{+} \ast K_{\sigma}' \rVert_{1}}{2\lambda_{+}} + \frac{p_{-} \lVert f_{-} \ast K_{\sigma}' \rVert_{1}}{2\lambda_{-}}.
\end{equation}
Recall that Young's convolution inequality establishes that
\begin{equation}
\label{eq:YoungConvolutionInq}
    \lVert h_{1} \ast h_{2} \rVert_{r} \leq \lVert h_{1} \rVert_{r_{1}} \lVert h_{2} \rVert_{r_{2}},
\end{equation}
whenever $\displaystyle \frac{1}{r_{1}} + \frac{1}{r_{2}} = \frac{1}{r} + 1$. Hence, by taking $r = r_{1} = r_{2} = 1$,
\begin{equation}
    \lVert \theta' \rVert_{1} \leq \frac{p_{+} \lVert f_{+} \rVert_{1} \lVert K_{\sigma}' \rVert_{1}}{2\lambda_{+}} + \frac{p_{-} \lVert f_{-} \rVert_{1} \lVert K_{\sigma}' \rVert_{1}}{2\lambda_{-}}.
\end{equation}
Since $f_{\pm}$ is a probability density function, we have that $\lVert f_{\pm} \rVert_{1} = 1$. Therefore, Lemma~\ref{Lemma:NormDerivativesKSigma1} implies that
\begin{equation}
    \lVert \theta' \rVert_{1} \leq \left(\frac{p_{+}}{\lambda_{+}} + \frac{p_{-}}{\lambda_{-}}\right) \frac{1}{\sqrt{2\pi}\sigma},
\end{equation}
as required.
\end{proof}

The following corollary provides an upper bound for the $L^{1}$-norm of $\theta''$.

\begin{corollary}
\label{Corollary:L1NormTheta2}
If $f_{\pm}$ is a probability density function, then
\begin{equation}
    \lVert \theta'' \rVert_{1} \leq \left(\frac{p_{+}}{\lambda_{+}} + \frac{p_{-}}{\lambda_{-}}\right) \frac{1}{\sigma^{2}} + \left(\frac{p_{+}^{2}}{\lambda_{+}^{2}} + \frac{p_{-}^{2}}{\lambda_{-}^{2}}\right) \frac{1}{8\sqrt{\pi}\sigma^{3}}.
\end{equation}
\end{corollary}

\begin{proof}
In Lemma~\ref{Lemma:DerivativesTheta} we prove that, for all $x\in\mathbb{R}$,
\begin{equation}
    2\theta''(x) = \left[\frac{g_{+}''(x)}{g_{+}(x)}-\left(\frac{g_{+}'(x)}{g_{+}(x)}\right)^{2}\right] - \left[\frac{g_{-}''(x)}{g_{-}(x)}-\left(\frac{g_{-}'(x)}{g_{-}(x)}\right)^{2}\right].
\end{equation}
By the triangle inequality, we have that
\begin{align}
    2 \lvert \theta''(x) \rvert &\leq \frac{\lvert g_{+}''(x) \rvert}{g_{+}(x)} + \frac{\lvert g_{+}'(x)\rvert^{2}}{g_{+}(x)^{2}} + \frac{\lvert g_{-}''(x) \rvert}{g_{-}(x)} + \frac{\lvert g_{-}'(x) \rvert^{2}}{g_{-}(x)^{2}}\\
    &\leq \frac{\lvert g_{+}''(x) \rvert}{\lambda_{+}} + \frac{\lvert g_{+}'(x)\rvert^{2}}{\lambda_{+}^{2}} + \frac{\lvert g_{-}''(x) \rvert}{\lambda_{-}} + \frac{\lvert g_{-}'(x) \rvert^{2}}{\lambda_{-}^{2}},
\end{align}
where the last inequality follows trivially from \eqref{eq:Defgsigmatheta}. Thus,
\begin{equation}
\label{eq:CorollaryTheta2Decomposition}
    \lVert \theta'' \rVert_{1} \leq \frac{\lVert g_{+}'' \rVert_{1}}{2\lambda_{+}} + \frac{\lVert g_{+}' \rVert_{2}^{2}}{2\lambda_{+}^{2}} + \frac{\lVert g_{-}'' \rVert}{2\lambda_{-}} + \frac{\lVert g_{-}' \rVert_{2}^{2}}{2\lambda_{-}^{2}}.
\end{equation}
From \eqref{eq:Defgsigmatheta}, it is immediate to see that $g_{+}'' = p_{+} (f_{+} \ast K_{\sigma})''$. Hence, the formula $ (h_{1} \ast h_{2})'' = h_{1} \ast h_{2}''$ implies that
\begin{equation}
    \lVert g_{\pm}'' \rVert_{1} = p_{\pm} \lVert f_{\pm} \ast K_{\sigma}'' \rVert_{1} \leq p_{\pm} \lVert f_{\pm} \rVert_{1} \lVert K_{\sigma}'' \rVert_{1},
\end{equation}
where we applied Young's convolution inequality \eqref{eq:YoungConvolutionInq} with $r = r_{1} = r_{2} = 1$. Since $f_{\pm}$ is a probability density function, we have that $\lVert f_{\pm} \rVert_{1} = 1$. Therefore, Lemma~\ref{Lemma:NormDerivativesKSigma2} implies that
\begin{equation}
\label{eq:CorollaryTheta2G2}
    \lVert g_{\pm}'' \rVert_{1} \leq \frac{2p_{\pm}}{\sigma^{2}}.
\end{equation}
Similarly, we have that
\begin{equation}
    \lVert g_{\pm}' \rVert_{2}^{2} = p_{\pm}^{2} \lVert f_{\pm} \ast K_{\sigma}' \rVert_{2}^{2} \leq p_{\pm}^{2} \lVert f_{\pm} \rVert_{1}^{2} \lVert K_{\sigma}' \rVert_{2}^{2},
\end{equation}
where we applied Young's convolution inequality \eqref{eq:YoungConvolutionInq} with $r = r_{2} = 2$ and $r_{1} = 1$. Thus, Lemma~\ref{Lemma:NormDerivativesKSigma1} implies that
\begin{equation}
\label{eq:CorollaryTheta2G1}
    \lVert g_{\pm}' \rVert_{2}^{2} \leq \frac{p_{\pm}^{2}}{4\sqrt{\pi}\sigma^{3}}.
\end{equation}
By plugging \eqref{eq:CorollaryTheta2G2} and \eqref{eq:CorollaryTheta2G1} in \eqref{eq:CorollaryTheta2Decomposition}, we conclude that
\begin{equation}
    \lVert \theta'' \rVert_{1} \leq \left(\frac{p_{+}}{\lambda_{+}} + \frac{p_{-}}{\lambda_{-}}\right) \frac{1}{\sigma^{2}} + \left(\frac{p_{+}^{2}}{\lambda_{+}^{2}} + \frac{p_{-}^{2}}{\lambda_{-}^{2}}\right) \frac{1}{8\sqrt{\pi}\sigma^{3}},
\end{equation}
as required.
\end{proof}

The following corollary provides an upper bound for the $L^{1}$-norm of $\theta'''$.

\begin{corollary}
\label{Corollary:L1NormTheta3}
If $f_{\pm}$ is a probability density function, then
\begin{align}
    \lVert \theta''' \rVert_{1} &\leq \left(\frac{p_{+}}{\lambda_{+}} + \frac{p_{-}}{\lambda_{-}}\right) \frac{5}{\sqrt{2\pi}\sigma^{3}} + \left(\frac{p_{+}^{2}}{\lambda_{+}^{2}} + \frac{p_{-}^{2}}{\lambda_{-}^{2}}\right) \frac{3\sqrt{3}}{8\sqrt{2\pi}\sigma^{4}} + \left(\frac{p_{+}^{3}}{\lambda_{+}^{3}} + \frac{p_{-}^{3}}{\lambda_{-}^{3}}\right) \frac{\sqrt{2}}{9\sqrt{\pi^{3}}\sigma^{5}}.
\end{align}
\end{corollary}

\begin{proof}
In Lemma~\ref{Lemma:DerivativesTheta} we prove that
\begin{align}
    2\theta''' &= \left[\frac{g_{+}'''}{g_{+}} - 3 \frac{g_{+}'g_{+}''}{g_{+}^{2}} + 2 \left(\frac{g_{+}'}{g_{+}}\right)^{3}\right] - \left[\frac{g_{-}'''}{g_{-}} - 3 \frac{g_{-}'g_{-}''}{g_{-}^{2}} + 2 \left(\frac{g_{-}'}{g_{-}}\right)^{3}\right].
\end{align}
By the triangle inequality, we have that
\begin{align}
    2\lvert \theta''' \rvert &\leq \frac{\lvert g_{+}''' \rvert}{g_{+}} + 3 \frac{\lvert g_{+}' g_{+}'' \rvert}{g_{+}^{2}} + 2 \frac{\lvert g_{+}'\rvert^{3}}{g_{+}^{3}} + \frac{\lvert g_{-}''' \rvert}{g_{-}} + 3 \frac{\lvert g_{-}' g_{-}'' \rvert}{g_{-}^{2}} + 2 \frac{\lvert g_{-}'\rvert^{3}}{g_{-}^{3}}\\
    &\leq \frac{\lvert g_{+}''' \rvert}{\lambda_{+}} + 3 \frac{\lvert g_{+}' g_{+}'' \rvert}{\lambda_{+}^{2}} + 2 \frac{\lvert g_{+}'\rvert^{3}}{\lambda_{+}^{3}} + \frac{\lvert g_{-}''' \rvert}{\lambda_{-}} + 3 \frac{\lvert g_{-}' g_{-}'' \rvert}{\lambda_{-}^{2}} + 2 \frac{\lvert g_{-}'\rvert^{3}}{\lambda_{-}^{3}},
\end{align}
where the last inequality follows trivially from \eqref{eq:Defgsigmatheta}. Thus,
\begin{align}
    \label{eq:CorollaryTheta3Decomposition} \lVert \theta''' \rVert_{1} &\leq \frac{\lVert g_{+}''' \rVert_{1}}{2\lambda_{+}} + \frac{3\lVert g_{+}' g_{+}'' \rVert_{1}}{2\lambda_{+}^{2}} + \frac{\lVert g_{+}' \rVert_{3}^{3}}{\lambda_{+}^{3}} + \frac{\lVert g_{-}''' \rVert_{1}}{2\lambda_{-}} + \frac{3\lVert g_{-}' g_{-}'' \rVert_{1}}{2\lambda_{-}^{2}} + \frac{\lVert g_{-}'\rVert_{3}^{3}}{\lambda_{-}^{3}}.
\end{align}
From \eqref{eq:Defgsigmatheta}, it is immediate to see that $g_{+}''' = p_{+} (f_{+} \ast K_{\sigma})'''$. Hence, the formula $ (h_{1} \ast h_{2})''' = h_{1} \ast h_{2}'''$ implies that
\begin{equation}
    \lVert g_{\pm}''' \rVert_{1} = p_{\pm} \lVert f_{\pm} \ast K_{\sigma}''' \rVert_{1} \leq p_{\pm} \lVert f_{\pm} \rVert_{1} \lVert K_{\sigma}''' \rVert_{1},
\end{equation}
where we applied Young's convolution inequality \eqref{eq:YoungConvolutionInq} with $r = r_{1} = r_{2} = 1$. Since $f_{\pm}$ is a probability density function, we have that $\lVert f_{\pm} \rVert_{1} = 1$. Therefore, Lemma~\ref{Lemma:NormDerivativesKSigma3} implies that
\begin{equation}
\label{eq:CorollaryTheta3G3}
    \lVert g_{\pm}''' \rVert_{1} \leq \frac{5\sqrt{2}p_{\pm}}{\sqrt{\pi}\sigma^{3}}.
\end{equation}
By H\"{o}lder's inequality, we observe that
\begin{equation}
    \lVert g_{\pm}' g_{\pm}'' \rVert_{1} \leq \lVert g_{\pm}' \rVert_{2} \lVert g_{\pm}'' \rVert_{2}.
\end{equation}
As before, we have that
\begin{equation}
    \lVert g_{\pm}'' \rVert_{2} = p_{\pm} \lVert f_{\pm} \ast K_{\sigma}'' \rVert_{2} \leq p_{\pm} \lVert f_{\pm} \rVert_{1} \lVert K_{\sigma}'' \rVert_{2},
\end{equation}
where we applied Young's inequality \eqref{eq:YoungConvolutionInq} with $r = r_{2} = 2$ and $r_{1} = 1$. Thus, Lemma~\ref{Lemma:NormDerivativesKSigma2} implies
\begin{equation}
    \lVert g_{\pm}'' \rVert_{2} \leq \frac{\sqrt{3}p_{\pm}}{2\sqrt{2}\sqrt[4]{\pi}\sigma^{5/2}}.
\end{equation}
The previous inequality and \eqref{eq:CorollaryTheta2G1} lead to
\begin{equation}
\label{eq:CorollaryTheta3G1G2}
    \lVert g_{\pm}' g_{\pm}'' \rVert_{1} \leq \frac{\sqrt{3}p_{\pm}^{2}}{4\sqrt{2\pi}\sigma^{4}}.
\end{equation}
Finally, we have that
\begin{equation}
    \lVert g_{+}' \rVert_{3}^{3} = p_{\pm}^{3} \lVert f_{\pm} \ast K_{\sigma}' \rVert_{3}^{3} \leq p_{\pm}^{3} \lVert f_{\pm} \rVert_{1}^{3} \lVert K_{\sigma}' \rVert_{3}^{3},
\end{equation}
where we applied Young's convolution inequality \eqref{eq:YoungConvolutionInq} with $r = r_{2} = 3$ and $r_{1} = 1$. Thus, Lemma~\ref{Lemma:NormDerivativesKSigma1} implies that
\begin{equation}
\label{eq:CorollaryTheta3G1}
    \lVert g_{+}' \rVert_{3}^{3} \leq \frac{\sqrt{2}p_{\pm}^{3}}{9\sqrt{\pi^{3}}\sigma^{5}}.
\end{equation}
By plugging \eqref{eq:CorollaryTheta3G3}, \eqref{eq:CorollaryTheta3G1G2} and \eqref{eq:CorollaryTheta3G1} in \eqref{eq:CorollaryTheta3Decomposition}, we conclude that
\begin{align}
    \lVert \theta''' \rVert_{1} &\leq \left(\frac{p_{+}}{\lambda_{+}} + \frac{p_{-}}{\lambda_{-}}\right) \frac{5}{\sqrt{2\pi}\sigma^{3}} + \left(\frac{p_{+}^{2}}{\lambda_{+}^{2}} + \frac{p_{-}^{2}}{\lambda_{-}^{2}}\right) \frac{3\sqrt{3}}{8\sqrt{2\pi}\sigma^{4}} + \left(\frac{p_{+}^{3}}{\lambda_{+}^{3}} + \frac{p_{-}^{3}}{\lambda_{-}^{3}}\right) \frac{\sqrt{2}}{9\sqrt{\pi^{3}}\sigma^{5}},
\end{align}
as required.
\end{proof}

Now we are in position to prove Theorem~\ref{Theorem:MainTheoremPostProcessingTheta}.

\begin{proof}[\textbf{Proof of Theorem~\ref{Theorem:MainTheoremPostProcessingTheta}}]
By Theorem~\ref{Theorem:CEtaL1NormsBound}, we have that
\begin{equation}
    C_{\theta} \leq \frac{2\sqrt{2}}{\sqrt{\pi}} \left(1 + \frac{1}{2}\log\left(\lVert\theta'\rVert_{1} \lVert\theta'''\rVert_{1}\right) - \log(\lVert\theta''\rVert_{1})\right) \lVert\theta''\rVert_{1}.
\end{equation}
Recall the definition of $N_{\alpha}$ in \eqref{eq:DefN1Sigma} -- \eqref{eq:DefN3Sigma}. Corollaries~\ref{Corollary:L1NormTheta1} -- \ref{Corollary:L1NormTheta3} imply that $\lVert\theta^{(\alpha)}\rVert_{1} \leq N_{\alpha}$ for every $\alpha\in\{1,2,3\}$. As a result,
\begin{align}
    \label{eq:MainTheoremPostProcessingTheta1} C_{\theta} &\leq \frac{2\sqrt{2}N_{2}}{\sqrt{\pi}} \left(1 + \frac{1}{2}\log\left(N_{1}N_{3}\right)\right) - \frac{2\sqrt{2}}{\sqrt{\pi}} \log(\lVert\theta''\rVert_{1}) \lVert\theta''\rVert_{1}.
\end{align}
Since $-\log(z)z \leq 1/e$ for all $z\in[0,\infty)$, the previous inequality implies that
\begin{equation}
    C_{\theta} \leq \frac{2\sqrt{2}}{e\sqrt{\pi}} + \frac{2\sqrt{2}N_{2}}{\sqrt{\pi}} \left(1 + \frac{1}{2}\log\left(N_{1}N_{3}\right)\right).
\end{equation}
It is straightforward to verify that $z \mapsto -\log(z)z$ is increasing over $[0,1/e]$. Thus, if $N_{2} \leq 1/e$, \eqref{eq:MainTheoremPostProcessingTheta1} implies that
\begin{equation}
    C_{\theta} \leq \frac{2\sqrt{2}N_{2}}{\sqrt{\pi}} \left(1 + \frac{1}{2} \log\left(\frac{N_{1}N_{3}}{N_{2}^{2}}\right)\right),
\end{equation}
as required.
\end{proof}

\section{Proof of Lemma~\ref{Lemma:StructuralPropertyHk}}
\label{Appendix:ProofLemmaStructuralPropertyHk}

\begin{proof}[\textbf{Proof of Lemma~\ref{Lemma:StructuralPropertyHk}}]
Observe that, without loss of generality, we can assume that
\begin{equation}
    \delta < \frac{1}{2} \min\left\{t_{2}-t_{1},\ldots,t_{k}-t_{k-1}\right\}.
\end{equation}
Since $\displaystyle \lim_{\zeta\to\infty} \tanh(\zeta) = 1$ and $\tanh(-\zeta) = - \tanh(\zeta)$, there exists $\zeta_{\epsilon}>0$ such that, for all $s\in\{-1,0,+1\}$,
\begin{equation}
\label{eq:StructuralPropertyHkDefzetaeps}
    \lvert s - \tanh(\zeta_{\epsilon}s) \rvert \leq \frac{\epsilon}{2}.
\end{equation}
Recall that $\phi$ satisfies that $\displaystyle \lim_{z\to\pm\infty} \phi(z) = \pm1$. Hence, there exists $z_{\epsilon}>0$ such that, for all $z \geq z_{\epsilon}$,
\begin{equation}
\label{eq:StructuralPropertyHkDefzeps}
    \lvert 1 - \phi(z) \rvert \leq \frac{\epsilon}{2k\zeta_{\epsilon}} \quad \text{and} \quad \lvert -1 - \phi(-z) \rvert \leq \frac{\epsilon}{2k\zeta_{\epsilon}}.
\end{equation}

Note that any $g\in\mathcal{S}_{k}$ can be written as
\begin{equation}
\label{eq:ProofStructuralPropertyHkAlternativeExpressiong}
    g(x) = s_{0} + \sum_{l=1}^{k} (s_{l}-s_{l-1}) \mathbbm{1}_{[t_{l},\infty)}(x),
\end{equation}
where $s_{l}\in\{-1,0,+1\}$ and $t_{1}<\cdots<t_{k}$. For such a $g\in\mathcal{S}_{k}$, let $h\in\mathcal{H}_{k}^{\phi}$ be the function defined by
\begin{equation}
    h(x) = \zeta_{\epsilon}\left(s_{0} + \sum_{l=1}^{k} (s_{l}-s_{l-1})\frac{\phi\left(\frac{z_{\epsilon}}{\delta}(x-t_{l})\right)+1}{2}\right).
\end{equation}
In the sequel we show that $h$ satisfies \eqref{eq:ThmStructurlProperty}.

Assume that $t_{j}+\delta \leq x \leq t_{j+1}-\delta$ for some $j\in\{0,\ldots,k\}$. In this case, \eqref{eq:ProofStructuralPropertyHkAlternativeExpressiong} implies that $g(x) = s_{j}$ and, as a result,
\begin{align}
    \lvert g(x) - \tanh(h(x)) \rvert &\leq \lvert s_{j} - \tanh(\zeta_{\epsilon}s_{j}) \rvert + \lvert \tanh(\zeta_{\epsilon}s_{j}) - \tanh(h(x)) \rvert\\
    \label{eq:ProofStructuralPropertyHkTriangleInq} &\leq \frac{\epsilon}{2} + \lvert \zeta_{\epsilon} s_{j} - h(x) \rvert,
\end{align}
where we used \eqref{eq:StructuralPropertyHkDefzetaeps} and the fact that $\tanh$ is $1$-Lipschitz. A straightforward manipulation shows that
\begin{align}
    h(x) &= \zeta_{\epsilon}s_{j} + \zeta_{\epsilon} \sum_{l=1}^{j} (s_{l}-s_{l-1}) \frac{\phi\left(\frac{z_{\epsilon}}{\delta}(x-t_{l})\right)-1}{2} + \zeta_{\epsilon} \sum_{l=j+1}^{k} (s_{l}-s_{l-1})\frac{\phi\left(\frac{z_{\epsilon}}{\delta}(x-t_{l})\right)+1}{2}.
\end{align}
In particular, we have that
\begin{align}
    \label{eq:StructuralPropertyHkReduction} \lvert \zeta_{\epsilon}s_{j} - h(x) \rvert &\leq \zeta_{\epsilon} \sum_{l=1}^{j} \left\lvert \phi\left(\frac{z_{\epsilon}}{\delta}(x-t_{l})\right)-1 \right\rvert + \zeta_{\epsilon} \sum_{l=j+1}^{k} \left\lvert \phi\left(\frac{z_{\epsilon}}{\delta}(x-t_{l})\right)+1 \right\rvert.
\end{align}
Note that, for all $l \leq j$,
\begin{equation}
    \frac{z_{\epsilon}}{\delta}(x - t_{l}) \geq \frac{z_{\epsilon}}{\delta}(x - t_{j}) \geq z_{\epsilon},
\end{equation}
and, for all $l \geq j+1$,
\begin{equation}
    \frac{z_{\epsilon}}{\delta} (x-t_{l}) \leq \frac{z_{\epsilon}}{\delta} (x-t_{j+1}) \leq -z_{\epsilon}.
\end{equation}
Therefore, \eqref{eq:StructuralPropertyHkDefzeps} and \eqref{eq:StructuralPropertyHkReduction} imply that
\begin{equation}
    \lvert \zeta_{\epsilon}s_{j} - h(x) \rvert \leq \frac{\epsilon}{2},
\end{equation}
and, as a result, \eqref{eq:ProofStructuralPropertyHkTriangleInq} becomes
\begin{equation}
    \lvert g(x) - \tanh(h(x)) \rvert \leq \epsilon,
\end{equation}
as required.
\end{proof}

\bibliographystyle{IEEEtran}
\bibliography{Bibliography}

\end{document}